\documentclass[oneside, 11pt, reqno]{amsart}
\usepackage[left=1in, right=1in, top=1in, bottom=1in]{geometry}

\usepackage[usenames,dvipsnames]{xcolor}

\usepackage{graphicx,enumitem}
\usepackage{multirow}
\newcommand{\figcaption}[1]{
	\addtocounter{figure}{1}
	{\center Figure \arabic{figure}. #1\\ }
	\addcontentsline{lof}{figure}{#1}
}
\DeclareMathOperator\Liq{Liq}
\DeclareMathOperator\sgn{sgn}
\DeclareMathOperator\Int{Int}
\DeclareMathOperator\Cl{Cl}
\newcommand{\argmax}{\operatorname{argmax}}
\newcommand{\sAS}{\sA(\tS)}
\newcommand{\sCS}{\sC(\tS)}
\newcommand{\eit}[1]{\EE\left[\int_0^{\infty}  #1\, dt \right] }
\newcommand{\ebs}{\eta_B,\eta_S}
\newcommand{\tPP}{\bar{\PP}}
\newcommand{\eittp}[1]{\EE^{\tPP}\left[\int_0^{\infty}  #1\, dt \right] }
\newcommand{\sx}{\hS^x}
\newcommand{\hpix}{\hat{\pi}^x}
\newcommand{\hvx}{\hat{V}^x}
\newcommand{\hcx}{\hat{c}^x}
\newcommand{\hkx}{\hat{\kappa}^x}
\newcommand{\hvpzx}{\hat{\vp}^{0,x}}
\newcommand{\hvpx}{\hat{\vp}^x}

\newcommand{\ox}{\overline{x}}    \newcommand{\ux}{\underline{x}}
 
\newcommand{\old}{\overline{\ld}} \newcommand{\uld}{\underline{\ld}}

\newcommand{\oy}{\overline{y}}    \newcommand{\uy}{\underline{y}}
\newcommand{\oS}{\overline{S}}    \newcommand{\uS}{\underline{S}}
\newcommand{\os}{\overline{s}}    \newcommand{\us}{\underline{s}}

\newcommand{\eps}{\varepsilon}
\newcommand{\ld}{\lambda}

\newcommand{\vp}{\varphi}
\newcommand{\vpz}{\vp^0}
\newcommand{\vpu}{\vp^{\uparrow}} \newcommand{\vpd}{\vp^{\downarrow}}
\newcommand{\Phixu}{\Phi^{x\uparrow}}\newcommand{\Phixd}{\Phi^{x\downarrow}}
\newcommand{\hvpz}{\hat{\vp}^0}
\newcommand{\hvp}{\hat{\vp}}
\newcommand{\hd}{\hat{\delta}}

\newcommand{\hSig}{\hat{\Sigma}}
\newcommand{\hGm}{\hat{\Gamma}}
\newcommand{\hth}{\hat{\theta}}
\newcommand{\halq}{\hat{\alpha}_q} \newcommand{\halz}{\hat{\alpha}_0}
\newcommand{\hbt}{\hat{\beta}}
\newcommand{\hgm}{\hat{\gamma}}

\newcommand{\tot}{\tfrac{1}{2}} % one half in a small-font frac

\newcommand{\abs}[1]{\left| #1 \right|} % absolute value
\newcommand{\set}[1]{\{#1\}} % curly brackets
\newcommand{\Bset}[1]{\Big\{#1\Big\}} % curly brackets

\newcommand{\sets}[2]{\set{#1\,:\,#2}} % a set with "such that"
\newcommand{\Bsets}[2]{\Bset{#1\,:\,#2}} % a set with "such that"
\newcommand{\ind}[1]{ {\mathbf 1}_{{#1}}} % convex-theoretic indicator of a set
\newcommand{\inds}[1]{ {\mathbf 1}_{\set{#1}}} % indicator of a curly set
\newcommand{\seq}[1]{\set{#1_n}_{n\in\N}} % a sequence indexed by n
\newcommand{\prfi}[1]{ \{ #1 \}_{t\in [0,\infty)}} 

\providecommand{\R}{} \renewcommand{\R}{{\mathbb R}}
 
\newcommand{\N}{{\mathbb N}}
\newcommand{\PP}{{\mathbb P}}
\newcommand{\EE}{{\mathbb E}}

\newcommand{\define}[1]{{\textbf{#1}}}

\newcommand{\ito}{It\^ o}
\newcommand{\Zitkovic}[1]{{\v Z}itkovi\'c}
\newcommand{\Sirbu}[1]{S\^\i rbu}

\newcommand\hc{{\hat{c}}}

\newcommand\tu{{\tilde{u}}}
\newcommand\hu{{\hat{u}}}

\newcommand\hx{{\hat{x}}}

\newcommand\sA{{\mathcal A}}

\newcommand\sB{{\mathcal B}}

\newcommand\bB{{\mathbb B}}

\newcommand\sC{{\mathcal C}}

\newcommand\sD{{\mathcal D}}

\newcommand\sE{{\mathcal E}}

\newcommand\bE{{\mathbb E}}

\newcommand\sL{{\mathcal L}}

\newcommand\sP{{\mathcal P}}

\newcommand\bR{{\mathbb R}}

\newcommand\sS{{\mathcal S}}

\newcommand\tS{{\tilde{S}}}
\newcommand\hS{{\hat{S}}}

\newcommand\bT{{\mathbb T}}

\newcommand\sU{{\mathcal U}}

\newcommand\sV{{\mathcal V}}

\newcommand\bW{{\mathbb W}}

\newcommand\ga{{g_{\alpha}}}
\newcommand\ba{{\beta_{\alpha}}}

\numberwithin{equation}{section}
\theoremstyle{plain}                % title and number in bold, text italic
\newtheorem{theorem}{Theorem}[section]
\newtheorem{lemma}[theorem]{Lemma}
\newtheorem{proposition}[theorem]{Proposition}

\theoremstyle{definition}           % title and number in bold, text normal
\newtheorem{definition}[theorem]{Definition}

\theoremstyle{remark}
\newtheorem{remark}[theorem]{Remark}

\usepackage{graphicx} 
 
\begin{document}

\author{Jin Hyuk Choi, Mihai \Sirbu{} and Gordan \Zitkovic{}}
\title[Optimal consumption with transaction costs]{Shadow prices and
  well-posedness in the problem of  optimal investment and consumption
  with transaction costs}

\maketitle

\begin{abstract} We revisit  the optimal investment and consumption
  model of Davis and Norman (1990) and Shreve and Soner (1994),
  following a shadow-price approach similar to that  of Kallsen and
  Muhle-Karbe (2010).  Making use of the completeness of the model
  without transaction costs, we reformulate and reduce the
  Ha\-mil\-ton-Jacobi-Bellman equation for this singular stochastic
  control problem to  a non-standard free-boundary problem for a
  first-order ODE with an integral constraint. Having shown that the
  free boundary problem has a smooth solution, we use it to construct
  the solution of the original optimal investment/consumption problem
  in a self-contained manner and without any recourse to the dynamic
  programming principle.  Furthermore,  we provide an explicit
  characterization of model parameters for which the value function is
  finite.  
\end{abstract}

\section{Introduction} Ever since the seminal work of Merton (see
\cite{Mert:69} and \cite{Mert:71}), the problem of dynamic optimal
investment and consumption occupied  a central role in mathematical
finance and financial economics. Merton himself, together with many of
the researchers that followed him, made the simplifying assumption of
\emph{no market frictions}: there are no transaction costs, borrowing
and lending occur at the same interest rate, the assets can be bought
and sold immediately in any quantity and at the same price (perfect
liquidity), etc.  Among those, transaction costs are (arguably) among
the most important and (demonstrably) the most studied.

\subsection{Existing work} The problem of optimal investment where
transactions cost are present has received (and continues to receive)
considerable attention.  Following the early work of Constantinides
and Magill \cite{ConMag:76}, Davis and Norman \cite{DavNor90}
considered a risky asset driven by the geometric Brownian Motion for
which proportional transactions costs are levied on each transaction.
These authors formulated the optimal investment/consumption problem as
a singular stochastic control problem, and approached it using the
method of dynamic programming. Very early in the game it has been
intuited, and later proved to varying degrees of rigor, that the
optimal strategy has the following general form: 
\begin{enumerate}
  \item The investor should not trade at all as long as his/her
	holdings stay within the so-called ``no-trade region''  - a wedge
	around the Merton-proportion line.  
  \item Outside the ``no-trade
	region' , the investor should trade so as to reach the no-trading
	region as soon as possible, and, then, adjust the portfolio in a
	minimal way in order not to leave it ever after. 
\end{enumerate}
Such a strategy first appeared in \cite{ConMag:76} and was later made
more precise in \cite{DavNor90}.  The analysis of \cite{DavNor90} was
subsequently complemented by that of Shreve and Soner \cite{ShrSon94}
who removed various technical conditions, and clarified the key
arguments using the technique of viscosity solutions. Still, even in
\cite{ShrSon94}, technical conditions needed to be imposed. Most
notably, the analysis there assumes that \emph{the problem is well
  posed}, i.e., that the value function is finite; no necessary and
sufficient condition for this assumption, in terms of the parameters
of the model, is given in \cite{ShrSon94}. In fact, to the best of our
knowledge, the present paper provides the first such characterization.

\medskip

More recently, Kallsen and Muhle Karbe \cite{KalMuh10} approached the
problem using the concept of  a \emph{shadow price}, first introduced
by \cite{JouKal95} and \cite{LamPhaSch98}.  Roughly speaking, the
shadow-price approach amounts to comparing the problem \emph{with}
transaction costs to a family of similar problems, but \emph{without}
transaction costs, whose risky-asset prices lie between the bid  and
ask prices of the original model.  The most unfavorable of these
prices is expected to yield the same utility as the original problem
where transaction costs are paid.  As shown in \cite{KalMuh10}, this
approach works quite well for the case of the logarithmic utility,
which admits an explicit solution of the problem without transaction
costs in a very general class of not-necessarily Markovian models. The
fact that the logarithmic utility is the only member of the CRRA
(power) family of utility functions with that property makes a direct
extension of their techniques seem difficult to implement.  Very
recently, and in parallel with our work, partial results in this
direction have been obtained by Herczegh and Prokaj \cite{hp2011}
whose approach (and the intuition behind it) is based on the
second-order nonlinear free-boundary HJB equation of \cite{ShrSon94},
and applies only to a rather restrictive range of parameters.

\subsection{Our contributions.} Our results apply to the model
introduced \cite{DavNor90} or  \cite{ShrSon94}, and deal with
\emph{general power-utility functions} and \emph{general values of the
  parameters}.  It is based on the shadow-price approach, but quite
different in philosophy and execution from that of either
\cite{KalMuh10} or \cite{hp2011}. Our contributions can be divided
into two groups:

\bigskip

\noindent \textit{Novel treatment and  proofs of, as well as insights
  into the existing results.} We provide a new and complete path to
the solution to the optimal investment/consumption problem with
transaction costs and power-type utilities.  Our approach, based on
the notion of the shadow price,  is fully self-contained, does not
rely on the dynamic programming principle and expresses all the
features of the solution in terms of a solution to a single,
constrained free-boundary problem for a one-dimensional first-order
ODE. This way, it is able to distinguish between various parameter
regimes which remained hidden under the more abstract approach of
\cite{DavNor90} and \cite{ShrSon94}. Interestingly, most of those
regimes turn out to be ``singular'', in the sense that our first-order
ODE develops a singularity in the right-hand side. While we are able
to treat them fully, those cases require a much more delicate and
insightful analysis. The results of both \cite{KalMuh10} and
\cite{hp2011} apply only to the parameter regimes where no singularity
is present.

\bigskip 

\noindent \textit{New results.}
One of the advantages of our approach is that it allows us to give an
explicit characterization of the set of model parameters for which the
optimal investment and consumption problem with transaction costs is
well posed. As already mentioned above, to the best of our knowledge,
such a characterization is new, and not present in the literature.

\medskip

\noindent Not only as another application, but also as an integral
part of our proof,  we furthermore prove that a shadow price exists
whenever the problem is well-posed.

\medskip

\noindent Finally, our techniques can be used to provide precise
regularity information about all of the analytic ingredients, the
value function being one of them.  Somewhat surprisingly, we observe
that in the singular case these are not always real-analytic, even
when considered away from the free boundary.

\subsection{The organization of the paper.} The set-up and the main
results are presented in Section 2. In Section 3 we describe the
intuition and some technical considerations leading to our
non-standard free-boundary problem. In Section 4, we prove a
verification-type result, i.e., show how to solve the singular control
problem, assuming that a smooth-enough solution for the free-boundary
equation is available. The proof of existence of such a  smooth
solution is the most involved part of the paper. In order to make our
presentation easier to parse, we split this proof into two parts.
Section 5 presents the main ideas of the proof, accompanied by
graphical illustrations. The rigorous proofs follow in Section 6.

\section{The Problem and the Main Results}

\subsection{The Market} We consider a model of a financial market in
which the price process $\prfi{S_t}$ of the risky asset (form
simplicity called the ``stock'') is given by \[ dS_t= S_t(\mu\, dt+
  \sigma\, dB_t),\ t\in [0,\infty) \text{ with } S_0>0.\] Here, $B$ is
a standard Brownian motion, and $\mu>0$ and $\sigma>0$ are constants
- parameters of the model.  The information structure is given by the natural saturated filtration generated by $B$. An economic agent starts with $\eta_S$
shares of the stock and $\eta_B$ units of an interestless bond and
invests in the two securities dynamically.  Transaction costs are not
assumed away, and we model them as proportional to the size of the
transaction. More precisely, they are determined by two constants
$\uld\in (0,1)$ and $\old>0$: one gets only $\uS_t=(1-\uld) S_t$ for
one share of the stock, but pays $\oS_t=(1+\old)S_t$ for it.

\subsection{Solvency and Admissible Strategies} 
We assume that the agent's initial position $(\eta _B, \eta _S)$ is \define{strictly solvent}, which means that it can be liquidated to a positive cash position. More precisely, we assume that $\Liq(\eta _B,\eta _S,\uS_0,\oS_0)> 0$ where
\begin{equation}
   \label{equ:def-liq}
\Liq(\vp^0,\vp,\us,\os)= \vp^0 + \vp^+ \us- \vp^- \os.  
\end{equation}

The agent's
\define{(consumption/trading) strategy} is  described by a triple
$(\vpz,\vp,c)$ of {\bf optional} processes 
%(with respect to thenatural augmentation of the filtration generated by $B$) 
such that
$\vp$ and $\vpz$ are right-continuous and of finite variation and $c$
is nonnegative and  locally integrable, a.s.  The processes $\vpz$ and
$\vp$ have the meaning of the amount of cash held in the money market
and the number of shares in the risky asset, respectively, while $c$
is the consumption rate.

In order to incorporate the potential initial jump we distinguish
between the initial values $(\varphi ^0_{0-}, \varphi _{0-})$ and the
values $(\varphi ^0_{0}, \varphi _{0})$ (after which the processes are
right-continuous).  This is quite  typical for optimal
investment/consumption strategies, both in frictional and frictionless
markets,  when the agent initially holds stocks, in addition to bonds.
In this spirit, we always assume that $(\varphi ^0_{0-}, \varphi
_{0-})=(\eta _B, \eta_S).$

 A strategy $(\vpz,\vp,c)$ is said to be \define{self-financing} if 
\begin{equation}
   \label{equ:self-fin}
   % \nonumber 
   \begin{split}
 \vpz_t = \vpz_{0-}- \int_0^t \oS_u d\vpu_u + \int_0^t \uS_u
d \vpd_u-\int_0^t c_u\, du,
   \end{split}
\end{equation}
 where $\vp= \vp_{0-}+\vpu - \vpd$ is the pathwise minimal
 (Hahn-Jordan) decomposition of $\vp$ into a difference of two
 non-decreasing adapted, right-continuous processes, \emph{with
   possible jumps at time zero}, as we assume that
 $\vpu_{0-}=\vpd_{0-}=0.$

  The integrals used in (\ref{equ:self-fin}) above, 
 with respect to the (pathwise Stieltjes) measures
  $d\vpu$ and $d\vpd$ characterized by
  $d\vpu((a,b])= \vpu(b)-\vpu(a)$, and
  $d\vpd((a,b])= \vpd(b)-\vpd(a)$,
  for $0\leq a < b <\infty$ together with 
 $d\vpu(\{0\})= \vpu(0)-\vpu(0-)=\vpu(0)$, and
   $d\vpd(\{0\})= \vpd(0)-\vpd(0-)=\vpd(0)$.

A self-financing strategy $(\vpz,\vp,c )$ is called
\define{admissible} if its position is always \define{solvent}, i.e., if
\begin{equation}
   \label{equ:liq}
   % \nonumber 
   \begin{split}
\Liq(\vpz_t,\vp_t,\uS_t,\oS_t)\geq 0, \text{ for all $t$, a.s. }
   \end{split}
\end{equation}
The set of
all admissible strategies with $\vpz_{0-}=\eta_B$ and $\vp_{0-}=\eta_S$ is
denoted by $\sA $, and the set of all $c$ such that
$(\vpz,\vp,c)\in\sA$ for some $\vpz$ and $\vp$ -
the so-called \define{financeable consumption processes} - is denoted by
$\sC$. 

\subsection{Utility functions}
For $p\in (-\infty,1)$, we consider the \define{utility function}
$U:[0,\infty)\to [-\infty,\infty)$ of the power (CRRA) type. It is defined
for $c\geq 0$ by
\[ U(c) = \begin{cases} \tfrac{1}{p} c^p, & c\ne 0, p\ne 0\\ \log(c), & c\ne 0, p=0,
\end{cases}\text{ and }\ U(0)= \begin{cases} 0, & p>0,\\ -\infty, & p\leq 0
\end{cases}
\]
Our
task is to analyze the optimal investment and consumption problem, with the value
\begin{equation}
   \label{equ:oc-prob}
   % \nonumber 
   \begin{split}
 u = \sup_{c\in \sC} \sU(c),
\text{ where }\  \sU(c)=\eit{e^{-\delta t} U(c_t)},
   \end{split}
\end{equation}
and $\delta>0$ stands for the (constant) \define{impatience rate}. As
part of the definition of $\sU$, we posit that $\sU(c)=-\infty$ unless
$\eit{e^{-\delta t} (U(c_t))^-}<\infty$.

\subsection{Consistent price processes} An \ito{}-process $\tS$ is
called a \define{consistent price (process)} if $\uS_t\leq \tS_t\leq
\oS_t$, for all $t\geq 0$, a.s.; the set of all consistent prices
is denoted by $\sS$. For each consistent price $\tS\in\sS$, and a solvent pair
of initial holdings $(\eta_B,\eta_S)\in\R^2$ such that $\Liq(\eta _B,\eta _S,\uS_0,\oS_0)\geq 0$, we define the set
$\sAS$ of \define{(frictionless) admissible} strategies
$(\vpz,\vp,c)$, as we would in the standard frictionless market where
the price-process is given by $\tS$. More precisely, for
$(\vpz,\vp,c)$ to belong to $\sAS$ it is necessary and sufficient that
the following three conditions hold
\begin{enumerate}
\item[(i)]  $\vpz, \vp$ and $c$ are
progressively measurable, $c_t\geq 0$,
a.s., for all $t\in [0,\infty]$, 
\item [(ii)] $\vpz_{0-}=\eta_B$ and $\vp_{0-}=\eta_S$, and
\item [(iii)] $V_t=\vpz_t + \vp_t \tS_t\geq 0$, for all $t\in [0,\infty)$, a.s,. and
  \begin{equation}
   \label{equ:V}
   % \nonumber 
   \begin{split}
	 V_t=V_{0-}+ \int_0^t \vp_u\, d\tS_u-\int_0^t c_u\, du, t\geq 0, \text{ a.s.}
   \end{split}
\end{equation}
\end{enumerate} 
The set of processes $\prfi{c_t}$
that appear as the third component of an element of
$\sAS$ will be denoted by $\sCS$, i.e.,
\[ \sCS=\sets{c}{\text{there exist }\vpz,\vp,\text{ such that }
  (\vpz,\vp,c)\in\sAS}.\] The elements of $\sCS$ can be interpreted as
the consumption processes financeable from the initial holding
$(\ebs)$ in the frictionless market modeled by $\tS$.  The intuition
that the presence of transaction costs can only reduce the collection
of financeable consumption processes can be formalized as in the
following easy proposition.
\begin{proposition}
\label{pro:C-in-CS}
$\sC\subseteq \sCS$, for each $\tS\in\sS$.
\end{proposition}
\begin{proof}
  For $c\in\sC$, let $(\vpz,\vp)$ be such that $(\vpz,\vp,c)\in
  \sA$. By the self-financing condition (\ref{equ:self-fin}),
  the fact that $\uS_t\leq \tS_t\leq \oS_t$ and integration by parts
  %(supported by (\ref{equ:left-St})
 (simplified by the fact that
  $\tS$ is continuous), we have
\begin{equation}
   \label{equ:by-parts}
   % \nonumber 
   \begin{split}
-\int_0^t c_u\, du
&= \vpz_t - \vpz_{0-}+ \int_0^t \oS_u d\vpu_u - \int_0^t \uS_u
d \vpd_u \geq \vpz_t - \vpz_{0-} + 
\int_0^t \tS_u\, d\vp_u\\ 
&= \vpz_t - \vpz_{0-} + \tS_t \vp_t - \tS_0 \vp_{0-} - \int_0^t \vp_{u}\, d\tS_u
   \end{split}
\end{equation}
Therefore, by the admissibility criterion (\ref{equ:liq}), we have
\begin{equation}
   \label{equ:pos}
   % \nonumber 
   \begin{split}
\eta_B+ \tS_0 \eta_S + \int_0^t \vp_{u}\, d\tS_u-\int_0^t c_u\, du
\geq \vpz_t +\tS_t \vp_t\geq 0. 
   \end{split}
\end{equation}
It remains to set $\tilde{\vp}=\vp$ and $\tilde{\vp}^0_t=
\eta_B+\int_0^t \tilde{\vp}_u\, d\tS_u - \int_0^t c_u\, du
- \tilde{\varphi}_t \tilde{S}_t+\eta_S \tilde{S}_0$,
 and
observe that (\ref{equ:pos}) directly implies (\ref{equ:V}). Thus,
$(\tilde{\vp}^0,\tilde{\vp},c)\in \sAS$.
\end{proof}
It will be important in the sequel to be able to check whether
an element of $\sCS$ belongs to $\sC$. It happens,
essentially, when a strategy that finances it ``buys'' only when
$\tS_t=\oS_t$ and ``sells'' only when $\tS_t=\uS_t$. A precise
statement is given in the following proposition.
\begin{proposition}
\label{pro:equally-important}
  Given $\tS\in\sS$, let $c\in\sCS$ be such that
  there exist processes $\vpz$ and $\vp$ such that
  \begin{enumerate}
  \item \label{ite:1-equally-important}
$(\vpz,\vp,c)\in \sAS$, 
  \item \label{ite:2-equally-important}
$\vp$ is a right-continuous process of finite variation, and
  \item \label{ite:3-equally-important}the Stieltjes measure on
    $[0,\infty)$ induced by $\vpu$ is carried by $ \set{\tS_t=\oS_t}$
    and that induced by $\vpd$ by $ \set{\tS_t=\uS_t}$.
  \end{enumerate}
Then, $c\in \sC$. 
\end{proposition}
\begin{proof}
  Let the triplet $(\vpz,\vp,c)\in\sAS$ satisfy the conditions of
  the proposition. In particular, we have
\begin{equation}
\nonumber 
   \begin{split}
  0  &= \eta_B + \eta_S \tS_0- \vpz_t-\vp_t\tS_t+ \int_0^t \vp_u\,
  d\tS_u-\int_0^t\, c_u\, du .
   \end{split}
\end{equation}
Thanks to  condition \eqref{ite:3-equally-important}, the
integration-by-parts formula and the self-financing property
(\ref{equ:self-fin}), it follows that
\begin{equation}
\nonumber 
   \begin{split}
  \vpz_t  &= \eta_B - \int_0^t \tS_t\, d\vp_u-\int_0^t\, c_u\, du
= 
\eta_B - 
\int_0^t \tS_t\, d\vpu_u +
\int_0^t \tS_t\, d\vpd_u -
\int_0^t c_u\, du\\
&=\eta_B - 
\int_0^t \oS_t\, d\vpu_u +
\int_0^t \uS_t\, d\vpd_u -
\int_0^t c_u\, du.
   \end{split}
\end{equation}
Hence,  $c\in\sC$. 
\end{proof}
\subsection{Shadow Problems}
\label{sse:shadow-pr}
For each consistent price process $\tS$, we define an
auxiliary optimal-consumption problem - called the
\define{$\tS$-problem}, with the value $u(\tS)$, by
\begin{equation}
   \label{equ:weak-duality}
   % \nonumber 
   \begin{split}
  u(\tS)=\sup_{c\in \sCS}\sU(c),\text{ so that } u \leq
  \inf_{\tS\in\sS} u(\tS),
   \end{split}
\end{equation}
where
$u$ is defined as in \eqref{equ:oc-prob}, 
and the inequality on the right is implied by 
 Proposition \ref{pro:C-in-CS}.
In words,  each consistent price $\tS$ affords at least as good an
investment opportunity as the original frictional market. 

It is in the heart of our approach to show that the
duality gap, in fact, closes, i.e., that the inequality in
\eqref{equ:weak-duality} becomes an equality; the worst-case shadow
problem performs no better than frictional one.
\begin{definition}
  A consistent price  $\tS$ is called a \define{shadow price}
  if $u=u(\tS)$. 
\end{definition}
The central idea of the present paper is to look for a shadow price as
the minimizer of the right-hand side of (\ref{equ:weak-duality})
viewed as a stochastic control problem. More precisely, we turn our
attention to a search for an optimizer in 
 the \define{shadow problem}:
\begin{equation}
   \label{equ:shadow}
   % \nonumber 
   \begin{split}
 \tu = \inf_{\tS\in\sS} u(\tS).
   \end{split}
\end{equation}
We start by tackling the shadow problem in a
formal manner and deriving an analytic object (a free-boundary
problem) related to its solution. Next, we show that this
free-boundary problem indeed admits a solution and use it to construct
the candidate shadow price. Finally, instead of showing that our
candidate is indeed an optimizer for \eqref{equ:shadow} and that
$u=\tu$, we use the
following direct consequence of Proposition
\ref{pro:equally-important}.
\begin{proposition}
\label{pro:when-shadow}
  Suppose that for $\tS\in\sS$ there exists a triplet
 $(\vpz, \vp,c)$ 
  such that
  \begin{enumerate}
  \item $(\vpz,\vp,c)$ satisfies conditions
    \eqref{ite:1-equally-important}, \eqref{ite:2-equally-important}
    and \eqref{ite:3-equally-important} of Proposition
    \ref{pro:equally-important}, and
  \item $\sU(c')\leq \sU(c)$, for all $c'\in \sCS$,
  \end{enumerate}
Then, $\tS$ is a shadow price.
\end{proposition}
\begin{remark}
  The route we take towards the existence of a shadow price may appear
to be somewhat circuitous. It is chosen so as to maximize the
intuitive appeal of the method and minimize (already formidable)
technical difficulties. 
\end{remark}

While the remainder of the paper is devoted to the implementation of
the above idea, we anticipate its final results here, for the convenience of
the reader. 
 An important by-product of our analysis is the  \emph{explicit
   characterization} of those  parameter values which result in a
 well-posed problem (the value function is finite).
 To the best of our knowledge, such a characterization is not present
 in the literature, and the finiteness of the value function is either
 assumed (as in \cite{ShrSon94}) or deduced from rather strong conditions 
 (as in \cite{KalMuh10}). 
\begin{theorem}\label{thm:well-posed}
Given  the environment parameters $\mu,\sigma\in
  (0,\infty)$ and the transaction costs $\uld\in (0,1)$, $\old>0$,
the following statements are equivalent:

\smallskip

\begin{figure}[h!]
  \noindent\begin{minipage}{0.47\textwidth}
%\item \label{ite:well-posed} 
  (1) The problem is well posed, i.e
  $$-\infty <u<\infty, $$ whenever
  $\Liq(\eta _B,\eta _S,\uS_0,\oS_0)>0$.

  \medskip

% \item \label{ite:parameters} 
  \noindent (2) The parameters of the model satisfy one
  of the following three conditions:\\
- $p\leq 0$,\\
- $0<p<1$ and $\mu<\sqrt{\frac{2\delta(1-p)\sigma^2}{p}}$,\\
- $0<p<1$, $\sqrt{\frac{2\delta(1-p)\sigma^2}{p}}\leq \mu <
\frac{\delta}p+\frac{(1-p)\sigma ^2}2$ and $$C(\mu,\sigma,
p,\delta)<\log(\tfrac{1+\old}{1-\uld}),$$
  where the function $C(\cdot,\cdot,\cdot,\cdot)$ is given by
  \eqref{ite:C-expression} in a closed form.
\end{minipage}
\hfill
\begin{minipage}{0.50\textwidth}
\begin{center}
  \includegraphics[width=\textwidth]{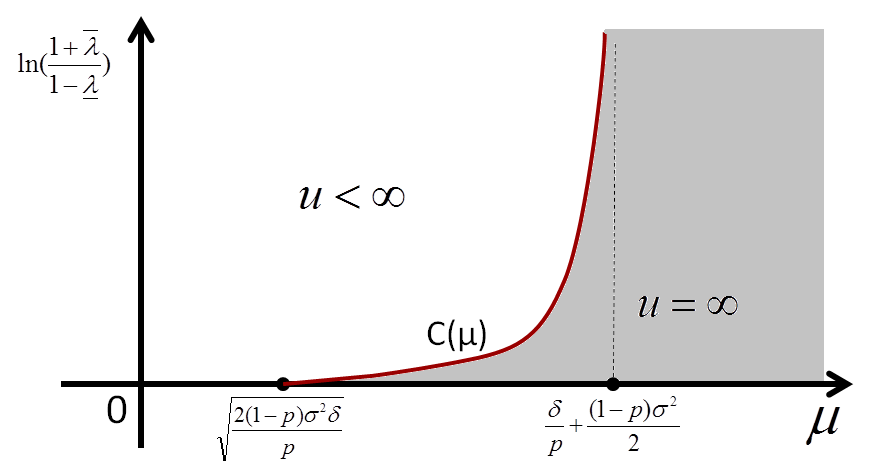}\\
  \caption{The well-posedness region.}
\end{center}
\end{minipage}
\end{figure}
\end{theorem}

%\begin{equation}
%\begin{split}\label{ite:C-expression}
%C(\mu,\sigma,p,\delta) = \int_0^{K(\mu,\sigma,p,\delta)} k \Big( \frac{l_u'(k)}{k-l_u(k)} - \frac{l_d'(k)}{k-l_d(k)} \Big) dk
%\end{split}
%\end{equation}
%where $l_u(k), l_d(k)$ are solutions to the quadratic equation $a(k) X^2 - b(k) X + c(k) = 0$ and $K(\mu,\sigma,p,\delta)$ is a smaller solution to $b(k)^2 = 4a(k)c(k)$  with
%\begin{equation}\label{K}
%\begin{split}
%\begin{array}{ll} a(k)=2p\delta(1+k), \\  b(k)= (2\delta + p(1-p)(2\mu-\sigma^2))k + 2p(1-p)\mu, \\ c(k)=(1-p)(2\mu+(p^2-1)\sigma^2)k + (1-p)^3 \sigma^2. \end{array},
%\end{split}
%\end{equation}
\begin{remark} For $\sigma =0$, the third condition in (2) above
 reduces to a well-known condition of Shreve and Soner. Indeed, the
 entire Section 12 in \cite{ShrSon94}, culminating in Theorem 12.2,
 p.~677, is devoted to the well-posedness problem with two bonds (i.e,
 with $\sigma=0$).
\end{remark}

As demonstrated by our second main result, 
the shadow-price approach not only allows us to fully characterize the
conditions under which a solution to the frictional
optimal investment/consumption problem exists, but it also sheds light on 
its form and regularity. 
\begin{theorem}
\label{thm:main}
 Given the %agent
% parameters $\ebs\in\R^2$, $p\in  (-\infty,1)$, $\delta>0$, the environment 
parameters $\mu,\sigma\in
  (0,\infty)$ and the transaction costs $\uld\in (0,1)$, $\old>0$, we
  assume that well-posedness conditions of
  Theorem \ref{thm:well-posed} hold. Then
\begin{enumerate}
%\item \label{ite:main-1} The problem is well posed, i.e
%$-\infty <u<\infty$
%\item \label{item-main-11} The parameters of the model satisfy 
%\begin{enumerate}
%\item either $p<0$ or
%\item $p\in [0,1)$, $\delta > \frac{p \mu ^2}{2 (1-p)\sigma ^2}$ 
%\end{enumerate}

\item \label{ite:main-3}

  There exist constants $\ux,\ox$ with $0< \ux < \ox$ and a function $g\in C^2[\ux,\ox]$ such that
 
  \begin{enumerate}
 % \item $\ux>0$ for $p\in (0,1)$ and $\ox<0$ for $p\in (-\infty,0]$, 
  \item $g'(x)>0$ for $x\in (\ux,\ox)$, and $g$ satisfies the equation 
\begin{equation}
   \label{equ:HJB-g}
   % \nonumber 
   \begin{split}
     \inf_{\Sigma,\theta\in\R} \Big( \tot \Sigma^2 \tfrac{x}{g'(x)} -
     \alpha_q(\Sigma,\theta) x - \beta(\theta) g(x) + \gamma(\theta)
	 \Big)=0,\ x\in (\ux,\ox),
   \end{split}
\end{equation}
where 
\begin{equation}%
    \label{equ:greeks}
    \begin{split}
& q=\tfrac{p}{1-p}, \ 
      \alpha_q(\Sigma,\theta) = \theta\sigma -\mu - \Sigma
      \Big(\tot \Sigma+\sigma-\theta(1+q) \Big),\\
      &\beta(\theta) =(1+q) \Big(\delta - \tot q  \theta^2\Big),\text{ and } 
      \gamma(\theta) =\begin{cases} \tot \theta^2, & p=0, \\
\sgn(p), & p\ne 0.
\end{cases}
\end{split}
\end{equation}
  \item \label{ite:main-4} the following boundary/integral conditions
	are satisfied:
	\begin{equation}\label{equ:integral-cond}
	 \begin{split}
	   g'(\ux+)=g'(\ox-)=0\text{ and }
	\int_{\ux}^{\ox} \tfrac{g'(x)}{x}\,
	dx=\log(\tfrac{1+\old}{1-\uld}).
	 \end{split}
	\end{equation}
  \item The function $h:[\ux,\ox]\to\R$, defined by 
\begin{equation}
   \label{equ:def-h}
   % \nonumber 
   \begin{split}
     h(x)=
     \begin{cases} (1-x) g'(x)+1,& p=0,\\
       q g(x) \left(g'(x)+1\right)-(q+1) x g'(x),& p\ne 0,\\
     \end{cases}
   \end{split}
\end{equation}
admits no zeros on $[\ux,\ox]$. 
\end{enumerate}
\item \label{ite:main-2}
For any $(\eta _B, \eta _S)$ such that 
\begin{equation}\label{strict:solvent}
\Liq(\eta _B,\eta _S,\uS_0,\oS_0)>0, 
\end{equation}
  there exists a shadow price $\prfi{\tS_t}$, of the form $\tS_t = S_t e^{f(X^{\hx}_t)}$, 
  where
  \begin{itemize}
	\item[-] $f(x)=\uy + \int_x^{\ox} \tfrac{g'(t)}{t}\, dt$, for
	  $x\in [\ux,\ox]$, 
	\item[-] the value of the constant $\hx$ is determined as in Proposition
	  \ref{prop:rx}, and
	\item[-]  $X^{\hx}$ is the unique
		solution of reflected SDE \eqref{equ:Skorokhod} with
		$X^{\hx}_0=\hx$.
  \end{itemize}
\item \label{ite:main-8} For any $(\eta _B, \eta _S)$ satisfying \eqref{strict:solvent},
 the value $u$ and an optimal investment/consumption strategy
 $(\hvpz,\hvp,\hc)$ for the main problem (\ref{equ:oc-prob}) are given by
$$u= \hu(\eta_B,\eta_S;\hx),\quad (\hvpz_t,\hvp_t,\hc_t)=(\hvp_t^{0,\hx},\hvp_t^{\hx},\hc_t^{\hx}), $$
where $\hx$ is defined in Proposition \ref{prop:rx}, and $\hu,
\hvp_t^{0,x},\hvp_t^{x}$ and $\hc_t^{x}$ in Lemma \ref{lem:complete-optimal}.
\end{enumerate}
\end{theorem}

\begin{remark}
In \eqref{equ:integral-cond}, if $(\ux,g(\ux))=P$ ($P$ is a singular point described in Section 5), then the condition $g'(\ux)=0$ can be violated. For this exceptional case, Proposition \ref{pro:Y} is still valid: More precisely, in the part (2) of Proposition \ref{pro:Y}, we need to show that $\frac{g'(X_t^x)}{X_t^x} d\Phi_t^{x}\equiv 0$. If $(\ux,g(\ux))=P$, in \eqref{equ:Skorokhod}, the drift is positive and the volatility is zero, thus, we conclude that $d\Phi_t^{x\uparrow} \equiv 0$.
%that the drift term is positive and the volatility term vanishes when $X_t^x=\ux$. Thus, the reflection is unnecessary when $X_t^x=\ux$, i.e., $d\Phi_t^{x\uparrow}\equiv 0$.
%$\hat{\Gamma}(\ux)=0$ and $\ux \hat{\beta}(\ux)-q\hat{\theta}(\ux)\hat{\Gamma}(\ux)>0$, so, conclude that $\Phi_t^{\uparrow}\equiv 0$.
\end{remark}

%\begin
\section{A heuristic derivation of a free-boundary problem}
\label{sec:heuristic}
The purpose of the present section is to provide a heuristic derivation of 
a free-boundary problem for a one-dimensional first-order ODE which will
later be used to construct a shadow process and the solution of our
main problem. With the fully rigorous verification coming later, we
often do not pay attention to integrability or measurability
conditions and formally push through many steps in this section.

We start by splitting the shadow problem \eqref{equ:shadow}
according to the starting value
of the process $\tS$:
\begin{equation}
   \label{equ:shadow-2}
   % \nonumber 
   \begin{split}
 \tu  = \inf_{s_0\in [(1-\uld) S_0, (1+\old) S_0]} \quad 
\inf_{\tS\in\sS, \tS_0=s_0} \quad \sup_{c\in \sC(\tS)} \sU(c).
   \end{split}
\end{equation}

One can significantly simplify the analysis of the above problem by
noting that, since each $\tS$ is a strictly positive \ito{}-process,
we can always choose processes $\Sigma=\Sigma(\tS)$ and
$\theta=\theta(\tS)$ such that
\[ d\tS_t= \tS_t (\sigma+\Sigma_t)\, (dB_t+\theta_t\, dt),\
\tS_0=s_0. \] 
It pays to 
pass to the logarithmic scale, and introduce the
process $Y_t=\log(\tS_t/ S_t)$, whose dynamics is given by
\begin{equation}
   \label{equ:Y-dynamics}
   % \nonumber 
   \begin{split}
 dY_t = \alpha_0(\theta_t,\Sigma_t)\, dt + \Sigma_t\, dB_t,
   \end{split}
\end{equation}
on the natural domain $Y_t\in [\uy,\oy]$. Here,
$\uy=\log(1-\uld)$,
$\oy=\log(1+\old)$
and the function $\alpha_0$ is given by \eqref{equ:greeks}
for $q=0$.
 This way, the family of consistent price
processes is  parametrized by the set
\[\sP=\sets{(y,\Sigma,\theta)}{ y\in [\uy,\oy], (\Sigma,\theta)\in
  \sP(y)},\] where $\sP(y)$ is the set of all pairs of regular-enough
processes $(\Sigma,\theta)$ such that the process $\prfi{Y_t}$, given
by (\ref{equ:Y-dynamics}) and starting at $Y_0= y$, stays in the
interval $[\uy,\oy]$, a.s. 

We note that the market modeled by $\tS$ is complete, and that, thanks
to the absence of friction, the agent with the initial holdings
$(\ebs)$ will achieve the same utility as the one who immediately
liquidates the position, i.e., the one with the initial wealth of
$\eta_B+ S_0 e^y \eta_S$. Therefore, the standard duality theory
suggests that
\begin{equation}
   \label{equ:dual}
   % \nonumber 
   \begin{split}
 \sup_{c\in \sCS} \sU(c)=\inf_{z>0} \Big(
(\eta_B+S_0 e^y \eta_S) z+ \sV(z \sE(-\theta\cdot B)) \Big),
   \end{split}
\end{equation}
 where
$\tS$ and $(y,\Sigma,\theta)$ are
related as above and
\begin{equation}
   \label{equ:sV-def}
   % \nonumber 
   \begin{split}
V(z)=\sup_{c>0} ( U(c) - c z ), \quad \sV(Z)=\eit{ e^{-\delta t} V(e^{\delta t}Z_t)}.
   \end{split}
\end{equation}
\begin{remark}
  The Legendre-Fenchel transform $V$ of $U$ admits an explicit and
  simple expression in the case of a power utility. Indeed, we have
  \[ V(z)= \begin{cases} \tfrac{1}{q} z^{-q}, & p\ne 0,\\
    -1-\log(z), & p=0,
  \end{cases}\] where $q=p/(1-p)$. The parameter $q$ is the negative
  of the conjugate exponent of $p$, i.e.,
  $\tfrac{1}{p}-\tfrac{1}{q}=1$ ($q=0$, for $p=0$) and this
  relationship will be assumed to hold throughout the paper without
  explicit mention.
\end{remark}
Consequently, if we combine \eqref{equ:shadow-2} and \eqref{equ:dual}, we 
 obtain the following equality:
\begin{equation}
\label{equ:split-pr}
   \begin{split}
     \tu = \inf_{ (y,z)\in [\uy,\oy]\times (0,\infty) } \Big((\eta_B+S_0
     e^y \eta_S) z+ \inf_{(\Sigma,\theta)\in \sP(y)} \sV(z
     \sE(-\theta\cdot B)) \Big).
   \end{split}
\end{equation}
The expression above is particularly convenient because it separates
the shadow problem into a stochastic control problem over $\sP(y)$,
and a (finite-dimensional) optimization problem over $y$ and $z$,
which can be solved separately.

\subsection{A dimensional reduction}
\label{sub:dim-reduce}
Thanks to homogeneity ($\log$-homogeneity for $p=0$) of the map
$z\mapsto \sV(zZ)$, a dimensional reduction is possible in the inner
control problem in (\ref{equ:split-pr}).  Indeed, with
$\hd=\delta (1+q)$, we have
\[ \sV(z \sE(-\theta\cdot B)) =
\begin{cases}
-\tfrac{2+\log(z)}{\delta} + 
\eit{e^{-\delta t} \Big(-\log(\sE(-\theta\cdot B)_t)\Big)}, & p=0,\\
\tfrac{z^{-q}}{q}\,  
\eit{e^{-\hd t} \sE(-\theta\cdot B)_t^{-q}}  , & p\ne 0. \\
\end{cases}
\]
Hence,
\begin{equation}
   \label{equ:hfu-form}
   % \nonumber 
   \begin{split}
 \tu = \inf_{y\in [\uy,\oy]} \begin{cases}
  \tfrac{1}{\delta}\Big(-1+\log\Big(\delta(\eta_B+S_0 e^y \eta_S)\Big) 
  + w(y) \Big), & p=0 \\
\tfrac{(\eta_B+S_0 e^y \eta_S)^p}{p} \abs{w(y)}^{1-p}, &p\ne 0,
\end{cases}
   \end{split}
\end{equation}
where
\[ w(y) = \inf_{(\Sigma,\theta)\in \sP(y)}
\begin{cases}
 \eit{\delta e^{-\delta t} \Big(-\log(\sE(-\theta\cdot B)_t)\Big)}, & p=0,\\
\sgn(p)\eit{e^{-\hd t} \sE(-\theta\cdot B)_t^{-q}}, & p\ne 0.
\end{cases}
\]
In the heuristic spirit of the present section, it will be assumed
that the processes of the form $\theta\cdot B$ and $\sE(q\theta\cdot
B)$ are (true) martingales so that the definition of the stochastic
exponential and the simple identity
\begin{equation}
   \label{equ:identity}
   % \nonumber 
   \begin{split}
\sE(-\theta\cdot B)^{-q} = \sE( q\theta\cdot B) \exp\Big( \tot q(1+q)
\int_0^{\cdot} \theta_u^2\, du\Big)
   \end{split}
\end{equation}
 can be used to  simplify the
expression for $w$ even further:
\begin{equation}
   \label{equ:w}
   % \nonumber 
   \begin{split}
 w(y) = \inf_{(\Sigma,\theta)\in \sP(y)}
\begin{cases}
 \tfrac{1}{2} 
\eit{e^{-\delta t} \theta^2_t}, & p=0,\\
\sgn(p)\eittp{e^{-\hd t}  e^{\tot q (1+q) \int_0^t \theta^2_u\, du}},& p\ne 0.
\end{cases}
   \end{split}
\end{equation}
Here, the measure\footnote{One should rather call $\tPP$ a cylindrical
measure, but, given the heuristic nature of the present section, we do
not pursue this distinction.} $\tPP$  is (locally) given by $d\tPP=\sE( q\theta
\cdot B)\, d\PP$.  By Girsanov's theorem the process
\begin{equation}
   \label{equ:BQ}
   % \nonumber 
   \begin{split}
     \bar{B}=B-\int_0^{\cdot}q\theta_u\, du
\end{split}
\end{equation}
 is (locally) a $\tPP$-Brownian motion and the dynamics of the 
process $Y$ can be conveniently written as 
\[ dY_t = \alpha_q(\theta_t,\Sigma_t)\, dt + \Sigma_t\, d\bar{B}_t.\]
The expression inside the infimum in (\ref{equ:w}) involves a
discounted running cost. Hence, it fits in the classical framework of
optimal stochastic control, and a formal HJB-equation can be written
down. We note that even though the process $\sE(-\theta\cdot B)$
appears in the original expression for $w$, the simplification in
(\ref{equ:w}) allows us to drop it from the list of state variables
and, thus, reduce the dimensionality of the problem. Indeed, the formal
HJB has the following form:
\begin{equation}
   \label{equ:HJB-w}
   % \nonumber 
   \begin{split}
 \inf_{\Sigma,\theta} \Big( \tot \Sigma^2 w''(y) + \alpha_q(\Sigma,\theta) w'(y) -
\beta(\theta) w(y) + \gamma(\theta) \Big)=0,\ y\in (\uy,\oy)
   \end{split}
\end{equation}
where the functions $\beta$ and $\gamma$ are defined in \eqref{equ:greeks}. 

In order to fully characterize the optimization problem, we need to
impose the boundary conditions at $\uy$ and $\oy$ to enforce the
requirement that $Y$ stay within the interval $[\uy,\oy]$. These
amount to turning off the diffusion completely and leaving only the
drift in the appropriate (inward) direction when $Y$ reaches the
boundary. Thanks to the form of the function $\alpha_q$ and the
equation (\ref{equ:HJB-w}), as well as the expectation that $w'$ be
bounded on $[\uy,\oy]$, we are led to
the boundary condition
\begin{equation}
   \label{equ:HJB-w-bd}
   % \nonumber 
   \begin{split}
 w''(\uy)=w''(\oy)= +\infty.
   \end{split}
\end{equation}
It will be shown in the following section that, in addition to the
annihilation of the diffusion coefficient, (\ref{equ:HJB-w-bd}) will
ensure that the drift coefficient $\alpha_q$ will indeed have the
proper sign of at the boundary.

\subsection{Shadow price as a strategy in a game}
\label{game}
By interpreting the
problem of shadow prices as a game, 
one can arrive to the two-point boundary problem \eqref{equ:HJB-w},
\eqref{equ:HJB-w-bd} without the use of duality. 

Let $\xi (y)=\eta _B+\eta_S S_0 e^y$ denote the
initial wealth of the investor who is not subject to transaction
costs,  for a 
  fixed $\eta_B, \eta_S$ and $y\in [\oy,\uy]$. Finding a shadow price
  now amounts to the following:
\begin{enumerate}
\item given $(\theta, \Sigma)$, solving the optimal investment problem for an investor with initial wealth
$\xi (y)$
i.e., maximizing 
the expected discounted utility
$$\EE[\int _0^{\infty} e^{-\delta t}U(c_t)].$$
over $(\varphi, c)$ (numbers of shares   of  $\tS$ and the rate of
consumption). 
\item  minimizing the obtained value over $(\theta, \Sigma)$, and
\item finally, further minimizing over  $y\in [\oy,\uy].$
\end{enumerate}
Up to the last minimization over $y$, the above defines 
a stochastic game with the value
$$v(\xi ,y):=\inf _{\theta , \Sigma}\sup _{\pi, c}\EE[\int _0^{\infty}
e^{-\delta t}U(c_t)].$$
The corresponding Isaacs equation with a two dimensional state 
$(\Xi,Y)$ and the initial condition $(\Xi _0, Y_0)=(\xi, y)$ scales as
\[ v(\xi,y)=\frac {\xi^p}{p} |w (y)|^{1-p}\text{ 
(compare to \eqref{equ:hfu-form}).}\] Consequently, it can be reduced
to a one-dimensional equation for $w(y)$, with $y\in [\oy,\uy]$, which
turns out to be exactly the boundary-value problem
\eqref{equ:HJB-w}, \eqref{equ:HJB-w-bd}.
Once the game is solved, we can find a shadow price by simply
minimizing the value $v(\xi (y),y)$ over $y\in [\oy,\uy]$. 

We believe that a similar approach - namely of 
rewriting the problem of optimal investment and consumption with
transaction costs as a game, through the use of consistent prices -
works well in more general situations, e.g., when multiple assets are
present.

\subsection{An order reduction} Finally, based on the fact that the
equation (\ref{equ:HJB-g}) is autonomous, we introduce an order-reducing
change of variable. With $w'$  expected to be increasing and continuous
on $[\uy,\oy]$, we define the function $g:[\ux,\ox]\to \R$, with
$\ox=-w'(\uy)$ and $\ux=-w'(\oy)$ by $w(y)=g(-w'(y))$. This transforms
the equation (\ref{equ:HJB-g}) into
\begin{equation}
\nonumber 
   \begin{split}
     \inf_{\Sigma,\theta} \Big( \tot \Sigma^2 \tfrac{x}{g'(x)} -
     \alpha_q(\Sigma,\theta) x - \beta(\theta) g(x) + \gamma(\theta)
     \Big)=0,
   \end{split}
\end{equation}
with (free) boundary conditions $g'(\ux)=g'(\ox)=0$ and
$\int_{\ux}^{\ox} \tfrac{g'(x)}{x}\, dx=\oy-\uy.$ 
%We note that, since non-trivial maximization in
%(\ref{equ:hfu-form}) is expected,
%the function $w$ should be decreasing, and, consequently,
The free boundaries $\ux$ and $\ox$ are expected
to be positive.

\section{Proof of the main theorem: verification}
\label{sec:verify}

We start the proof of our main Theorem \ref{thm:main} with a verification
argument which establishes the implication
$\eqref{ite:main-3}\implies \eqref{ite:main-2}$. After that, in Lemma
\ref{lem:complete-optimal} and Proposition \ref{prop:rx}, we show 
\eqref{ite:main-8}.

Let us assume, therefore, that a triplet $(\ux,\ox,g)$, as in part
\eqref{ite:main-3} of Theorem \ref{thm:main}, is given (and fixed for
the remainder of the section), and that the function $h$ is defined as
in \eqref{equ:def-h}. Let $\hth: [\ux,\ox]\to \R$ and
$\hSig:[\ux,\ox]\to \R$ be the formal optimizers of (\ref{equ:HJB-g}),
i.e.,
\begin{equation}
   \label{equ:optimizers}
   % \nonumber 
   \begin{split}
%\begin{cases}
%\hth(x) = \sigma x/ h(x),\  \hSig(x) = -\sigma (1-x) g'(x)/h(x), & p=0 \\
%\hth(x) = -\sigma (1-p) (q
%   g'(x)-1) x/  h(x),\ \hSig(x) = 
%   -\sigma (q g(x)-x) g'(x)/ h(x), & p\ne 0.
% \end{cases}\\
\begin{cases}
\hth(x) = \frac{\sigma x}{ h(x)},\  \hSig(x) = -\frac{\sigma (1-x) g'(x)}{h(x)}, & p=0 \\
\hth(x) = -\frac{\sigma(1-p) x(q
   g'(x)-1)}{h(x)},\ \hSig(x) = 
   -\frac{\sigma (q g(x)-x) g'(x)}{h(x)}, & p\ne 0.
 \end{cases}
   \end{split}
\end{equation}
 Similarly, let $\halq,\halz, \hbt, \hgm : [\ux,\ox]\to \R$ be the
 compositions of the functions $\alpha_q$, $\alpha_0$, $\beta$ and
 $\gamma$ of \eqref{equ:greeks}  with $\hth$ and $\hSig$. Using the
 explicit formulas above, one readily checks that function $\hGm(x)=-
 \tfrac{x}{g'(x)} \hSig(x)$ admits a Lipschitz extension to
 $[\ux,\ox]$.

 While the equation (\ref{equ:HJB-g}) can be written in a more
 explicit way - which will be used extensively later - for now we
 choose to keep its current variational form. We do note, however, the
 following useful property of the function $g$:
\begin{proposition}
For all $x\in (\ux,\ox)$ with $g'(x)\ne 0$, we have
\begin{equation}
   \label{equ:envelope}
   % \nonumber 
   \begin{split}
 \tot \hSig^2(x) \tfrac{d}{dx} \left( \tfrac{x}{g'(x)}
\right)  - \halq(x) - g'(x) \hbt(x)=0.
   \end{split}
\end{equation}
\end{proposition}
\begin{proof}
  The equation (\ref{equ:envelope}) follows either by direct
  computation (using the explicit formulas \eqref{equ:optimizers} for
  $\hSig$ and $\hth$ above) or the appropriate version of the Envelope
  Theorem (see, e.g., Theorem 3.3, p.~475 in \cite{GinKey02}), which
  states, loosely speaking, that we ``pass the derivative inside the
  infimum'' in the 
 equation
  (\ref{equ:HJB-g}). 
\end{proof}
\subsection{Construction of the state processes}\label{XY}
The family of processes $\prfi{X^x_t}$, $x\in [\ux,\ox]$, defined in
this section, will play the role of state processes in the
construction of the shadow-price process below. Thanks to the
Lipschitz property of the function $h$,  for each $x\in
[\ux,\ox]$ there exists a unique solution $\Big(
\prfi{X^x_t},\prfi{\Phi^x_t}\Big)$ of the following reflected
(Skorokhod-type) SDE
\begin{equation}
\label{equ:Skorokhod}
\left\{   \begin{split}
     dX^x_t &= \Big( 
X^x_t \hbt(X^x_t)
-q \hth(X^x_t) \hGm(X^x_t)
\Big)
\, dt + \hGm(X^x_t)\,
     dB_t+ d\Phi^x_t,\\
X^x_0&=x\in [\ux,\ox].
   \end{split}
\right.
\end{equation}
Here, $\Phi$ is the ``instantaneous inward reflection'' term for the
boundary $\set{\ux,\ox}$, i.e., a continuous  process of finite variation
whose pathwise Hahn-Jordan decomposition $(\Phixu, \Phixd)$ satisfies
\[ 
\Phixu_t = \int_0^t \inds{X^x_u=\ux}\, d\Phixu_u,\text{ and } \Phixd_t
= \int_0^t \inds{X^x_u=\ox}\, d\Phixd_u,\ t\geq 0.\] The reader is
referred to \cite{Sko61} for a more detailed discussion of various
possible boundary behaviors of diffusions in a bounded interval, as
well as the original existence and uniqueness result 
 \cite[pp.~269-274]{Sko61}) for \eqref{equ:Skorokhod}.

 For $x\in [\ux,\ox]$, we define the function $f:[\ux,\ox]\to\R$ by
 $f(x)=\uy+\int_x^{\ox} \tfrac{g'(\xi)}{\xi}\, d\xi$ and the process
 $\prfi{Y^x_t}$ by $Y^x_t= f(X^x_t)$. 
  In relation to the heuristic discussion of Section
 \ref{sec:heuristic}, we note that $f$ plays the (formal) role of the
 inverse of the derivative $w'$. Moreover, the process $Y^x$ has the following properties:
 \begin{proposition}
\label{pro:Y}
 For $x\in [\ux,\ox]$, we have
 \begin{enumerate}
 \item \label{ite:Y-1}
 $Y^x_t\in [\uy,\oy]$, for all $t\geq 0$, a.s., and
 \item \label{ite:Y-2}
$Y^x_0=f(x)$ and $dY^x_t= \halz(X^x_t)\, dt
   + \hSig(X^x_t)\, dB_t$.
 \end{enumerate}
 \end{proposition}
\begin{proof}
  Property \eqref{ite:Y-1} follows from the definition of the function
  $f$ and the assumption (c) of part $\eqref{ite:main-3}$ of Theorem
  \ref{thm:main}.  For \eqref{ite:Y-2}, \ito{}'s formula reveals the
  following dynamics of $Y^x$:
\begin{multline}
  \nonumber d Y^x_t= \Big( -g'(X^x_t) \hbt(X^x_t)- q\hSig(X^x_t)
  \hth(X^x_t)+ \tot \hSig^2(X^x_t)\tfrac{d}{dx} \big( \tfrac{x}{g'(x)}
  \big)\Big|_{x=X^x_t} \Big)\, dt+ \hSig(X^x_t)\,
  dB_t - \tfrac{g'(X^x_t)}{X^x_t}\, d\Phi^x_t.
\end{multline}
The identity \eqref{equ:envelope} allows us to simplify the above expression to 
\[ d Y^x_t= \halz(X^x_t)\, dt + \hSig(X^x_t)\,
dB_t - \tfrac{g'(X^x_t)}{X^x_t}\, d\Phi^x_t.\] Finally, since $g'$
vanishes on the boundary, the singular term disappears and we obtain 
the second statement. 
\end{proof}
\subsection{A stochastic representation for the function $g$}
For notational convenience, we define \[ W^x_t = \begin{cases}
  -\delta \log\Big(\sE\big(-\hth(X^x)\cdot B\big)_t\Big), & p=0,\\
  \sgn(p) \sE(-\hth(X^x) \cdot B)_t^{-q}, & p\ne 0.
\end{cases}
\]

\begin{proposition} \label{pro:rep-g}
For $x\in [\ux,\ox]$, we have 
\begin{equation}
   \label{equ:repr-g}
   % \nonumber 
   \begin{split}
g(x)= \EE[\int_0^{\infty} e^{-\hd t} W^x_t\, dt]. 
   \end{split}
\end{equation}
\end{proposition}
\begin{proof}
Using the equation (\ref{equ:HJB-g}),  relation
(\ref{equ:envelope}) and \ito{}'s formula, we can derive the following
dynamics for the process $g(X^x_t)$:
\begin{equation}%\label{}
    \nonumber 
    \begin{split}
      d g(X^x_t) &= 
\Big(\hbt(X^x_t) g(X^x_t) - \hgm(X^x_t) \Big)\, dt +g'(X^x_t) \hGm(X^x_t)\, d\bar{B}_t\\
    \end{split}
\end{equation}
where $\bar{B}$ is given by (\ref{equ:BQ}) with $\theta_t=\hth(X^x_t)$. On the other hand, if we 
set  $\rho^x_t=e^{-\int_0^t \hbt(X^x_u)\, du}$ and 
$H^x_t = \int_0^t \rho^x_u \hgm(X^x_u) \, du+ \rho^x_t g(X^x_t)$, we obtain
that 
\[ dH^x_t = \rho^x_t g'(X^x_t) \hGm(X^x_t)\, d\bar{B}_t.\] Girsanov's
theorem (applicable thanks to the boundedness of $\hth$) implies that
$\bar{B}$ is a Brownian motion on $[0,t]$, under the measure
$\tPP_t$, defined by $d\tPP_t=\sE(q \hth\cdot B)_t\, d\PP$. Therefore,
\[
\EE^{\tPP_t}\left[\rho^x_t g(X^x_t)\right] + \EE^{\tPP_t}\left[
  \int_0^t \rho^x_u \hgm(X^x_u)\, du\right] = \EE^{\tPP_t}[
H_t]=\EE^{\tPP_t}[H_0]=g(x),\] where the boundedness of the integrands
was used to do away with the  stochastic integrals with respect to $\bar{B}$.
The exponential identity (\ref{equ:identity}) now implies that
\begin{equation}
   \label{equ:g-expr}
   % \nonumber 
   \begin{split}
 g(x) = \EE^{\tPP_t}[\rho^x_t g(X^x_t)]+
\EE[\int_0^t e^{-\hd u}W^x_u \, du]+
\begin{cases}
0 ,& p\ne 0\\
\tot e^{-\delta t} \EE[ \int_0^t \hth(X^x_s)^2\, ds], & p=0\\
\end{cases}
   \end{split}
\end{equation}

For $p=0$, we can use the fact that $\hth$ and $g$ are bounded to
conclude that 
\[ \EE^{\tPP^t}[\rho^x_t g(X^x_t)]= \EE[ e^{-\delta t} g(X^x_t)] \to
0\text{ and }e^{-\delta t} \EE[ \int_0^t \hth(X^x_s)^2\, ds] \to 0.\]
These two limits can now easily be combined with \eqref{equ:g-expr} to yield
\eqref{equ:repr-g}. 

To deal with the case $p>0$, we note that non-negativity of $g$ and
$W^x_t$ in \eqref{equ:g-expr} implies that 
\begin{equation}
   \label{equ:fin-p1}
   % \nonumber 
   \begin{split}
\int_0^{\infty} e^{-\hd t}
\EE[W^x_t]\, dt<\infty.
   \end{split}
\end{equation}
Moreover, with
$\abs{g}_{\infty}=\sup_{x\in [\ux,\ox]} \abs{g(x)} $, we have
\[\EE^{\tPP_t}[\rho^x_t g(X^x_t)]\leq \abs{g}_{\infty}
\EE^{\tPP^t}[\rho^x_t]=\abs{g}_{\infty} e^{-\hd t} \EE[ W^x_t].\]
Therefore, it is enough to observe that $e^{-\hd t_n}
\EE[W^x_{t_n}]\to 0$, along sequence $\seq{t}$ with $t_n\to\infty$
which exists thanks to (\ref{equ:fin-p1}).

For $p<0$, the fact that $\rho^x_t \leq e^{-\hd t}$ implies that
$\EE^{\tPP_t}\left[\rho^x_t \abs{g(X^x_t)}\right] \leq
\abs{g}_{\infty}e^{-\hd t} \to 0$, which, in turn, 
together with \eqref{equ:g-expr}, implies \eqref{equ:repr-g}.
\end{proof}
\subsection{The Shadow Market}
For $x\in [\ux,\ox]$, we define the process $\sx_t= S_t e^{Y^x_t}$ and observe that, by
\ito{}'s formula, it admits the following dynamics: 
\begin{equation}
   \label{equ:sx-dynamics}
   % \nonumber 
   \begin{split}
 d\sx_t = \sx_t\Big(\sigma+\hSig(X^x_t)\Big)\Big( \hth(X^x_t)\, dt+
dB_t\Big),\ \sx_0= S_0 e^{f(x)}.
   \end{split}
\end{equation}
 The goal of this subsection is to show that $\sx_t$ is a
shadow price for the appropriate choice of the initial value $x\in
[\ux,\ox]$. Note that $\sigma+\hSig(X^x_t)$ is bounded away from $0$, by the properties of $g$.

Recall from subsection \ref{sse:shadow-pr} that $\hu(\sx)$ is the
value of the optimal consumption problem $\sU(c)\to\max$ in the
(frictionless) financial market driven by $\sx$ for an agent with the
initial holding $\eta_B$ in the bond, and $\eta_S$ in the stock (the
$\sx$-problem).  The following lemma, which describes the optimal
investment/consumption policy that achieves the maximum, will play a key role in
the proof of the shadow property of the process $\sx$. To simplify the
notation, we introduce the following shortcuts:
\[ \xi(x)= \eta_B+S_0 e^{f(x)}\eta_S,\ \Pi(x) = \begin{cases} 
x, & p=0,\\ \tfrac{x}{q g(x)}, & p\ne 0,
\end{cases}\ \text{ and } \ 
K(x) = \begin{cases}
  \delta & p=0, \\ \tfrac{1}{\abs{g(x)}}, & p\ne 0.\\
\end{cases}
\]
\begin{lemma}
\label{lem:complete-optimal}
  For $x\in [\ux,\ox]$ and the initial positions $(\eta_B,\eta_S)$ with
  $\eta_B+S_0 e^{f(x)}\eta_S\geq 0$, we have
\begin{equation}%
\label{equ:hu}
    \begin{split}
 \hu(\eta_B,\eta_S;x) = \begin{cases}
\tfrac{1}{\delta} \Big(-1+ \log(\delta \xi(x)) + g(x) \Big), & p=0,\\
\tfrac{1}{p} \xi(x)^p \abs{g(x)}^{1-p}, & p\ne 0.
\end{cases}
\end{split}
\end{equation}
Moreover, with the processes $\prfi{\hpix_t}$, $\prfi{\hkx_t}$ and
$\prfi{\hvx_t}$ defined by
\begin{equation}
   \label{equ:pi-c-V}
   % \nonumber 
   \begin{split}
 \hpix_t= \Pi(X^x_t),\quad \hkx_t= K( X^x_t), \quad
\hvx_t = \xi(x) \sE\Big(
\int_0^{\cdot} \tfrac{\hpix_u}{\sx_u}\, d\sx_u - \int_0^{\cdot}
\hkx_u\, du \Big)_t,
   \end{split}
\end{equation}
the optimal strategy $(\hvpzx, \hvpx, \hcx)$ for the $\sx$-problem 
is given for $t\geq 0$ by
\begin{equation}
   \label{equ:optimal}
\hcx_t = \hvx_t \hkx_t,\qquad 
\hvpzx_t =
\hvx_t(1-\hpix_t)\quad \text{ and }\quad \hvpx_t=\tfrac{\hvx_t \hpix_t}{\sx_t}.
\end{equation}
\end{lemma}
\begin{proof}
The standard complete-market duality
theory (see, e.g.,  Theorem 9.11, p.~141 in \cite{KarShr98}) implies that
$$ \hu(\eta_B,\eta_S;x) = \inf_{z>0} \left( (\eta_B+S_0 e^{f(x)} \eta_S) z +
  \sV\Big(z\sE\big(-\hth(X^x)\cdot B\big)\Big)\right) $$
where dual functional $\sV$ is as in (\ref{equ:sV-def}). Furthermore, following
the computations that lead to (\ref{equ:hfu-form}) in subsection
\ref{sub:dim-reduce}, and using the representation of Proposition
\ref{pro:rep-g}, we get \eqref{equ:hu}

Once the form of the value function $\hu$ has been determined, it is a
routine computation  derive the
expressions for the optimal investment/consumption strategy. Indeed,
let the processes $\prfi{\hpix_t}$ be given by
\begin{equation}
   \label{equ:inner-pi}
   % \nonumber 
   \begin{split}
 \hpix_t=
\tfrac{ \hth(X^x_t)}{(1-p)\big( \sigma+\hSig(X^x_t)\big)}
+\begin{cases}
 0, & p=0,\\
-\frac{X^x_t}{g(X^x_t)} 
\frac{ \hSig(X^x_t)}{ \sigma+\hSig(X^x_t)}, & p\ne 0,
\end{cases}
   \end{split}
\end{equation}
 and $\prfi{\hkx_t}$ and $\prfi{\hvx_t}$  as in the
statement. Then, one readily checks that the triplet $(\hvpzx,\hvpx,
\hcx)$ given by \eqref{equ:optimal} is an optimal
investment/consumption strategy.

Finally, the equality between the form \eqref{equ:inner-pi} and the
simpler one given in \eqref{equ:pi-c-V} in the statement follows by direct
computation where one can use the explicit formulas for the functions $\hSig$ and $\hth$ from \eqref{equ:optimizers}. 
\end{proof}
\begin{proposition}\label{prop:rx}
  Let $(\eta_B,\eta_S)$ be an admissible initial wealth, i.e., such
  that $\Liq(
  \eta_B,\eta_S,(1-\uld) S_0, (1+\old) S_0)\geq 0$. For the function
  $r:[\ux,\ox]\to\R$, given by $r(x) = \eta_S S_0 e^{f(x)} (1-\Pi(x)) -
  \eta_B \Pi(x)$, let the constant $\hx\in [\ux,\ox]$ be
  defined by
\[ \hx =
\begin{cases}
\ox, & r(x)>0\text{ for all } x\in [\ux,\ox]\\
\ux, & r(x)<0\text{ for all } x\in [\ux,\ox]\\
\text{a solution to } r(x)=0, & otherwise.
\end{cases}
\]
Then $\hS=\hat{S}^{\hx}$ is a
shadow price.
\end{proposition}
\begin{remark} The three possible cases in Proposition \ref{prop:rx} relate to whether the initial condition is outside the no-transaction region (above or below) or inside it. It is easy to check that $\hx$ minimizes the value
$\frac{\xi (x)^p}{p}|g(x)|^{1-p}$ for 
$\xi (x)=\eta _B+\eta _S S_0 e^{f(x)},$ as mentioned in Subsection \ref{game}.
\end{remark}
\begin{proof}
  The idea of the proof is to show that the triplet $(\hvpz, \hvp,
  \hc)$ of Lemma \ref{lem:complete-optimal} satisfies the conditions
  of Proposition \ref{pro:when-shadow}.  Since $\hc$ is the optimal
  consumption process, it will be enough to show that conditions
  \eqref{ite:2-equally-important} and \eqref{ite:3-equally-important}
  of Proposition \ref{pro:equally-important} hold. The expression
  \eqref{equ:optimal} implies that the processes $\hvpz$ and $\hvp$
  are continuous, except for a possible jump at $t=0$.

  Let us, first, deal with the jump at $t=0$. The conditions
  \eqref{ite:2-equally-important} and \eqref{ite:3-equally-important}
  of Proposition \ref{pro:equally-important} at $t=0$ translate into
  the following equality:
  \[ \hvpz_{0+} -\eta_B +(1+\old)S_0 (\vp_{0+}-\eta_S)^+ - (1-\uld)S_0
  (\vp_{0+}-\eta_S)^- =0,\]
which, after \eqref{equ:optimal} is used, becomes
\begin{equation}
   \label{equ:at-t=0}
   % \nonumber 
   \begin{split}
	 e^{f(x)} r(x) = G( r(x) ),\text{ where }
G(x) = (1-\uld) x^+ - (1+\old) x^-.
   \end{split}
\end{equation}
If $r(x)=0$ admits a solution $\hx\in [\ux,\ox]$, then
$x=\hx$ clearly satisfies the equation (\ref{equ:at-t=0}). On the
other hand, if $r(x)\ne 0$, for all $x\in [\ux,\ox]$, then by
continuity, either $r(x)>0$, for all $x\in [\ux,\ox]$ or $r(x)<0$, for
all $x\in [\ux,\ox]$. Focusing on the first possibility (with the
second one being similar) we note that in this case $G(r(x)) = (1-\uld)
r(x)$, and so, if we pick $\hx=\ox$, we get $e^{f(\hx)}r(\hx)=(1-\uld)
r(\hx)= G(r(\hx))$.

Next, we deal with the trajectories of the processes $\hvpz$ and
$\hvp$ for $t>0$.  It is a matter of a tedious but entirely
straightforward computation (which can be somewhat simplified by
passing to the logarithmic scale and using the identities
(\ref{equ:HJB-g}) and (\ref{equ:envelope})) to obtain the following
dynamics:
\begin{equation}
   \label{equ:rxt}
   % \nonumber 
   \begin{split}
     d\hvp_t = \tfrac{\hvp_t}{X^x_t}\, d\Phi^x_t.
     \end{split}
 \end{equation}
 Thanks to the fact that $\Phi^x$ is a finite-variation process which
 decreases only when $X^x_t=\ox$ (i.e., $\sx_t=\uS$) and 
 increases only when $X^x_t=\ux$ (i.e., $\sx_t=\oS$), the 
 conditions \eqref{ite:2-equally-important} and
 \eqref{ite:3-equally-important} of Proposition
 \ref{pro:equally-important} hold. 
\end{proof}

\section{Main ideas behind the proof of  existence for the free-boundary problem}
\label{sec:ODE}
Having presented a verification argument in the previous section, we
turn to the analysis of the (non-standard) free-boundary problem
\eqref{equ:HJB-g}, \eqref{equ:integral-cond}. 
We start by remarking that that
the equation (\ref{equ:HJB-g}) simplifies to the form
\[ g'(x) = L(x,g(x)),\text{ where } L(x,z) =\frac{P(x,z)}{Q(x,z)},\]
and where the second-order polynomials
 $P(x,z)$ and $Q(x,z)$ are given by 
\begin{equation}
\nonumber 
   \begin{split}
P(x,z)&=
\begin{cases}
-2q\delta z^2 + 2p(\mu x +\sgn(p))z - (1-p)^2 \sigma^2 x^2, &p\neq 0\\
-2\delta z - \sigma^2 x^2 + 2\mu x, & p=0\\
\end{cases}\\
Q(x,z)&=
\begin{cases}
-P(x,z) + (p\sigma^2-2\hd)xz + 2(\mu-(1-p)\sigma^2)x^2 + 2\sgn(p)x , &p\neq 0\\
(1-x)(2\delta z + (\sigma^2-2\mu)x), & p=0\\
\end{cases}\\
   \end{split}
\end{equation}
The existence proof is based on  
a geometrically-flavored analysis of the equation  \eqref{equ:HJB-g},
where the  
curves $\bT$ and $\bB$, given by 
\begin{equation}
\begin{split}
  \bT&= \Bsets{(x,z) \in (0,\infty)\times \R}{ P(x,z)=0}\\
  \bB&= \Bsets{(x,z) \in (0,\infty) \times \R}{ Q(x,z)=0},
\end{split}
\end{equation}
play a prominent role. 
Many cases need to be considered, but we always proceed according to the 
following \define{program}:
\begin{enumerate}
\item First, we note that the boundary conditions $g'(\ux)=g'(\ox)=0$ amount to 
$$(\ux, g(\ux)), (\ox, g(\ox))\in \mathbb{T}.$$
\item Then, for a fixed $(\alpha, z(\alpha))\in \mathbb{T}$
  we solve the ODE
$g'(x)=L(x,g(x))$ with initial condition 
$g(\alpha)=z(\alpha)$ and let it evolve  to the right (if possible) until meeting again the curve $\mathbb{T}$ at the $x$-intercept $\beta _{\alpha}$. We therefore obtain a solution $g_{\alpha}:[\alpha, \beta _{\alpha}]\rightarrow \mathbb{R}$ satisfying
$$g _{\alpha }'(\alpha)=g_{\alpha }'(\beta _{\alpha})=0.$$
If $P=Q=0$ on some point along the way, only continuity
is required there.
\item Finally, we vary the parameter $\alpha$ to meet the integral
  condition $\int_{\ux}^{\ox} \tfrac{g'(x)}{x}\, dx = \log(\tfrac{1+\old}{1-\uld}).$
\end{enumerate}

In order to give some intuition for the technicalities that follow,
let us consider, for a moment, the "degenerate" frictionless case
$\uld=\old=0$.
For fixed $\mu, \sigma, p$, and $\delta$, the absence of transaction costs
suggests a trivial solution with $\ux=\ox$.
In addition, the point $ (\ox, g(\ox))\in \mathbb{T}$   is expected to 
have the highest possible  $z$-coordinate. Indeed, larger values of
$g$ translate, as we saw in Lemma \ref{lem:complete-optimal}, to larger
expected utilities. If such a point exists 
we call it the \define{North pole}, denote it by $N$ and its
$x$-coordinate by $x_N>0$.
In that case, furthermore, 
 the curve $\mathbb{T}$ decomposes into two parts $\mathbb{W}$
 (West of North) and $\mathbb{E}$ (East of North) so that
$$\mathbb{T}=\mathbb{W}\cup \{N\}\cup \mathbb{E}.$$
\begin{remark}
  \label{north-pole}
  It turns out that:
\begin{enumerate}
\item The curve $\mathbb{T}$ has a North pole, if and only if
  $u<\infty$ when $\uld=\old =0$.
\item When the North pole does exist:
  \begin{enumerate}
\item if $\uld=\old =0$ then $\ux =\ox=x_N$ and $(\ux, g(\ux))=N$, and
\item if $\uld+\old>0$,  we expect
  $(\ux,g(\ux))\in \mathbb{W}$ and $(\ox,g(\ox))\in \mathbb{E}.$
\end{enumerate}
\end{enumerate}
\end{remark}
%While the above remark is mere heuristics for the moment, it will become a fully rigorous statement. We wanted to make it explicit to justify the large number of (technically different) cases below.
Before we go ahead, we note that the quantities
$\frac{2\delta}p$ and $(1-p)\sigma ^2$ together with their geometric
and arithmetic means play a special role. In fact, they deserve their
own notation:
\[
  G=G( \sigma, p, \delta)=\sqrt{\tfrac{2\delta (1-p)\sigma
	  ^2}{p}},\quad
  A=A(\sigma, p, \delta)=\tfrac {\delta}p+\tfrac{(1-p)\sigma ^2}2.\]
Another quantity that will play a role in the analysis is the {\em
  Merton proportion} 
\[ \pi=\pi (\mu, \sigma, p)=\tfrac {\mu}{(1-p)\sigma ^2},\]
for an investor in a frictionless market, with the power utility.
The last thing we need to do before we delve deeper into the analysis of
various cases, is to introduce a suitable notation for the \define{singular} points,
i.e., the points $(x,z)\in (0,\infty)\times \R$ with
$P(x,z)=Q(x,z)=0$. The explicit expressions for $P$ and $Q$ above yield
immediately that, in general, there are three($p\neq0$) or two($p=0$) solutions to $P=Q=0$ in
$\R^2$,
two($p\neq0$) or one($p=0$) of which lie on the $z$ axis (and, therefore, do not count as
singular points). The other one,
denoted by $P$ is the unique singular point and will be quite
important in our analysis. It has 
coordinates
\[ x_P = \begin{cases}  \frac{\sgn(p)}{A-\mu}, &
	p\ne 0,\\ 1, & p =0, \end{cases}\quad \text { and } 
  \quad  z_P=\begin{cases} \frac{1}{q}x_P, & p \ne 0, \\ \tfrac{ 2\mu -
	  \sigma^2}{2 \delta },& p=0,\end{cases}\]
which clearly degenerate for $A\leq \mu$, $p\ne 0$;
in those cases, we set $(x_P, z_P) = (\infty,\infty)$.

\medskip

We are now ready to start differentiating between several (technically
different) cases which are chosen,
roughly, according to the following criteria:
(1) whether the risk aversion is \emph{low} ($0<p<1$) or
\emph{high} ($p\leq 0$), \ (2) 
whether the ``North pole'' exists, and \ 
(3) the sign of $\pi-1$.

\subsection{Low risk aversion $0<p<1$} In this case 
investor is less risk averse than the log-investor, and it is the only 
case when well-posedness may fail.
We further separate it into several sub-cases:

\medskip

\emph{ - Sub-case a): $\mu<G$.}\ \ 
For these particular values of parameters, the problem turns out to be well posed.
The reason is simple: the value function of the {\em frictionless}
version is finitely-valued here.
 The curve $\mathbb{T}$ is (a portion of) an
ellipse, and, as such, it obviously has a "North pole", in agreement
with Remark \ref{north-pole}.

Let $E$ denote the  most
right-ward point (East) and by $x_E$ its $x$-coordinate, so that
$0<x_N<x_E$.
Taking into account Remark \ref{north-pole} and the fact that
$\mathbb{T}$ is an ellipse, we expect to find a solution $(\ux,\ox,g)$
of the free boundary which satisfies
$\ux<\ox\leq x_E.$ As described in the outline of our program above,
we ``evolve'' the solution from the initial point to the right, as long as we can. More
precisely, 
we consider a maximal (with respect to the domain) $C^2$-solution of
the initial-value problem
$g_{\alpha}'(x) = L(x,g_{\alpha}(x))$, $g'(\alpha)=0$ with the property
that $P(x,g_{\alpha}(x))\geq 0$, i.e., such that the curve $\ga$ stays on
the inside of $\bT$. 
It turns out that the domain of this
solution is of the form $[\alpha,\beta_{\alpha}]$, for some
$\beta_{\alpha}\in [\alpha,x_E]$, and that the following statements
hold:
\begin{enumerate}
  \item  The map $\alpha  \mapsto \int_{\alpha }^{\beta _{\alpha}}
  \frac{g_{\alpha}'(x)}{x} dx$ is continuous and strictly decreasing
  on $(0,x_N)$, and 
\item $\lim_{\alpha  \searrow 0} \int_{\alpha }^{\beta _{\alpha}}
  \frac{g_{\alpha }'(x)}{x} dx =\infty$ while
$\lim_{\alpha  \nearrow x_N} \int_{\alpha }^{\beta _{\alpha}}
\frac{g_{\alpha}'(x)}{x} dx = 0$.
\end{enumerate}
It follows immediately that a unique $\alpha$, such that $g_{\alpha}$
solves the free-boundary problem
\eqref{equ:HJB-g},\eqref{equ:integral-cond} exists. The major
difficulty in the analysis is the fact that, for a given $\alpha$, the
maximal solution $g_{\alpha}$ may encounter the singularity $P$ on its
trajectory 
(see Figure 3.~below). 

An important tool here turns out the be the
so-called \define{containment curve}, i.e., a function
$\tau:(0,\infty)\to \R$ such
that
\begin{itemize}[itemindent=-2em]
 \item $\ga$ cannot hit $\tau$ before it hits $\bT$, and
	\item $\ga$ must hit $\tau$ before it hits $\bB$.
 \end{itemize}
 It serves a two-fold purpose here. First, it restricts the possible
 values the function $\ga$ can take and makes sure that it
 either does not intersect the (singular) curve $\bB$ at all, or that it encounters
 it only at the point $P$.
The shaded area
$\Omega_0$ in Figure 2.~below depicts the region of the plane the graph $\Gamma_{\alpha}$
is restricted to lie in, under various conditions on the problem
parameters. 
Second, when the singular point $P$ indeed happens to lie on the graph
$\Gamma_{\alpha}$, a well-constructed
containment curve $\tau$ provides crucial information about the
behavior of $\ga$ in a neighborhood of $P$. 
Whether or not $P$ falls on the graph 
\[ \Gamma_{\alpha} = \sets{ (x,\ga(x))}{ x\in
	[\alpha,\beta_{\alpha}]}.\]
of $\ga$ depends on the
values of the parameters. In particular, it depends on the relative
position of the points $E$, $P$ and $N$. The lead actor turns out to
be the Merton proportion $\pi=\pi(\mu,\sigma,p)$, and the following
three cases need to be distinguished (see Figure 2., below):
\begin{enumerate}[leftmargin=0pt,itemindent=1.8em]
  \item $\pi<1$ :  $P\in \bE$ and $P\notin\Gamma_{\alpha}$, with the
	relative positions of $P$, $E$ and $N$, further determined by
	the sign of $\tfrac{2 \delta}p- (1-p)\sigma^2$
\item $\pi=1$ : Here, $P=N$ and $\ba=x_P$ for any $\alpha$. 
  \item $\pi>1$:  In this case, $P\in \bW$. Furthermore,  $P\in \Gamma_{\alpha}$ if and
	only if $\alpha\leq x_P$. 
\end{enumerate}

 \ \\[1ex]

\begin{center}
$
  \begin{array}{cccc}
\includegraphics[width=3.7cm]{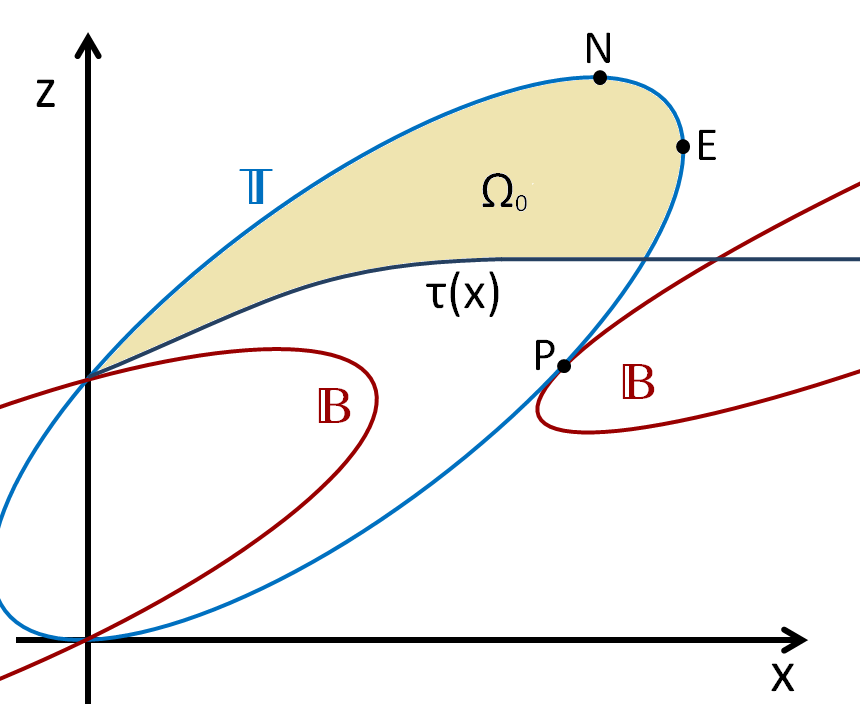} & \ 
   \includegraphics[width=3.7cm]{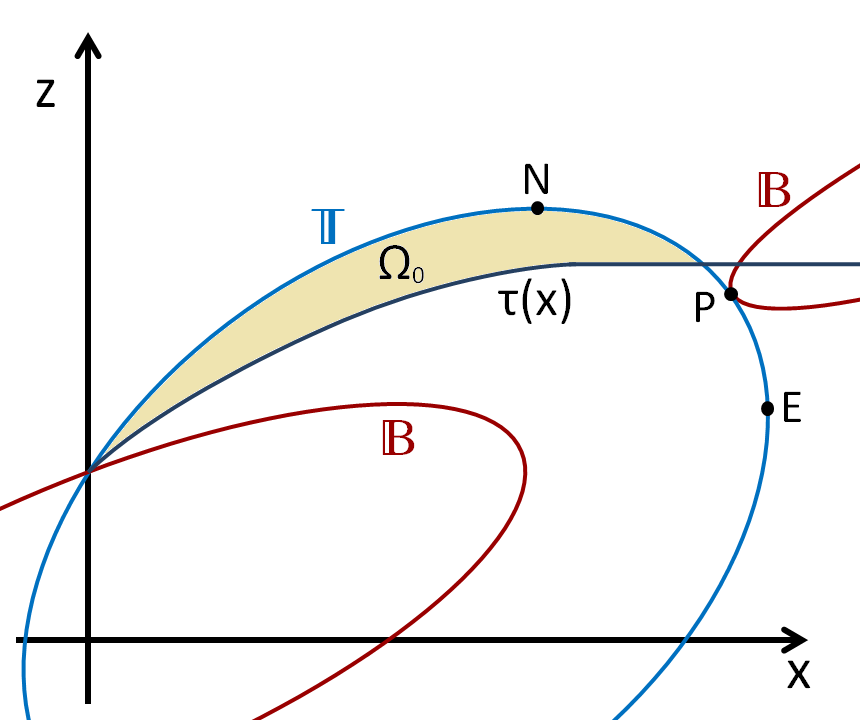} &
 \includegraphics[width=3.7cm]{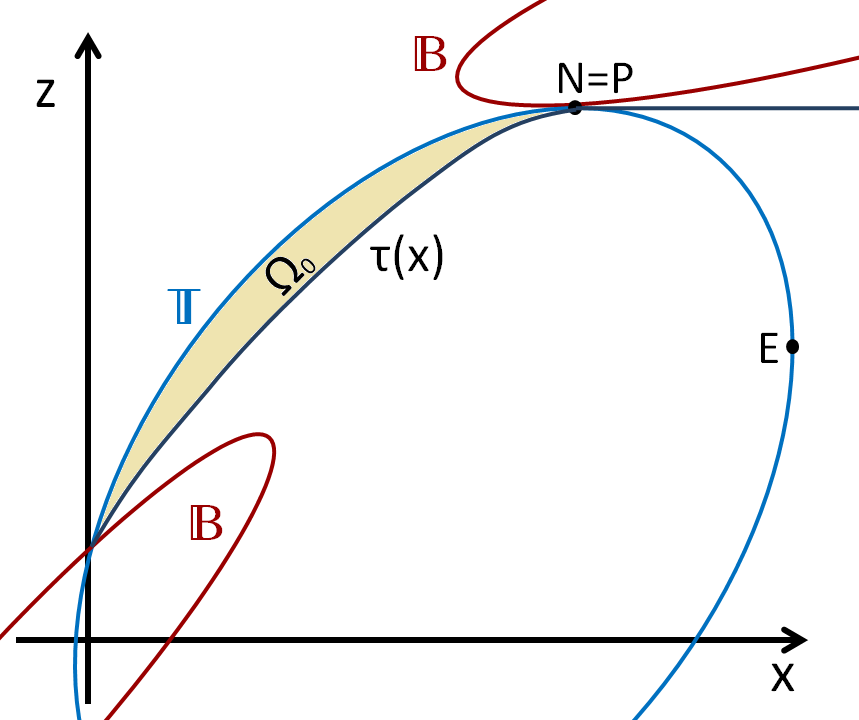} &  
\includegraphics[width=3.7cm]{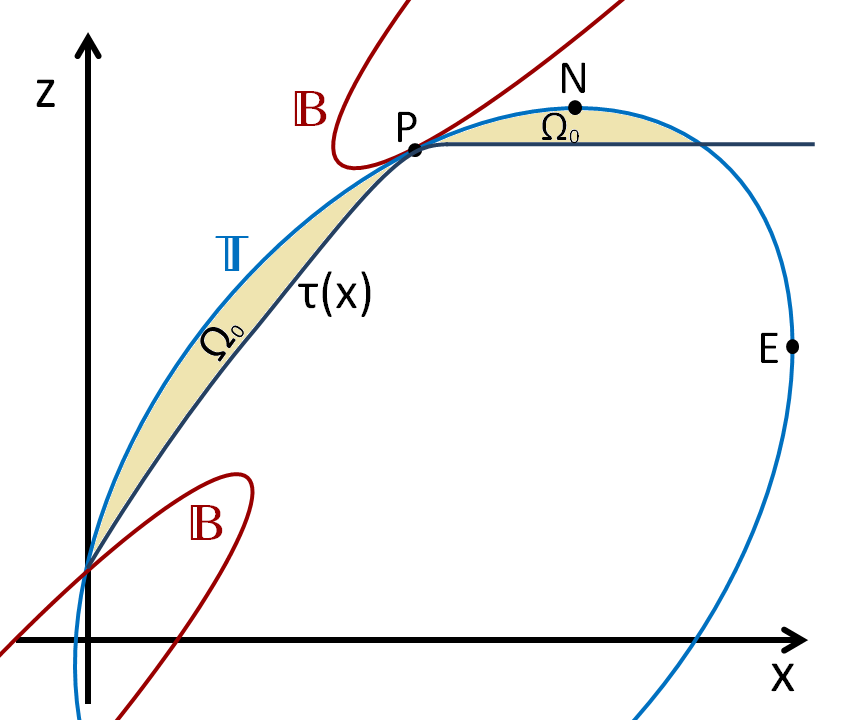}\\
\pi<1, \ \tfrac{\delta}{p}<\tfrac{(1-p)\sigma^2}{2} & 
\pi<1, \ \tfrac{\delta}{p}>\tfrac{(1-p)\sigma^2}{2} &
 \pi=1  
& \pi>1 
 \end{array}$\\[2ex]
 \figcaption{  $0<p<1$, $\mu<G$.}
 \end{center}
 
 \bigskip

 Figure 3.~below shows some of the possible shapes the graph
$\Gamma_{\alpha}$ can take, under a representative choice of parameter regimes.

 \ \\

\begin{center}$
\begin{array}{cccc}
\includegraphics[width=3.7cm]{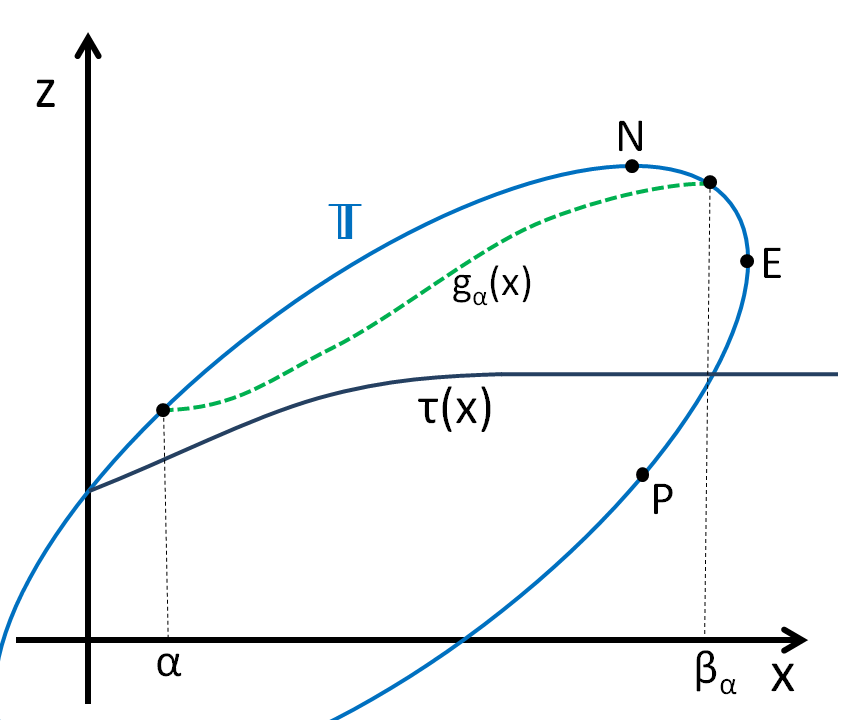} &
  \includegraphics[width=3.7cm]{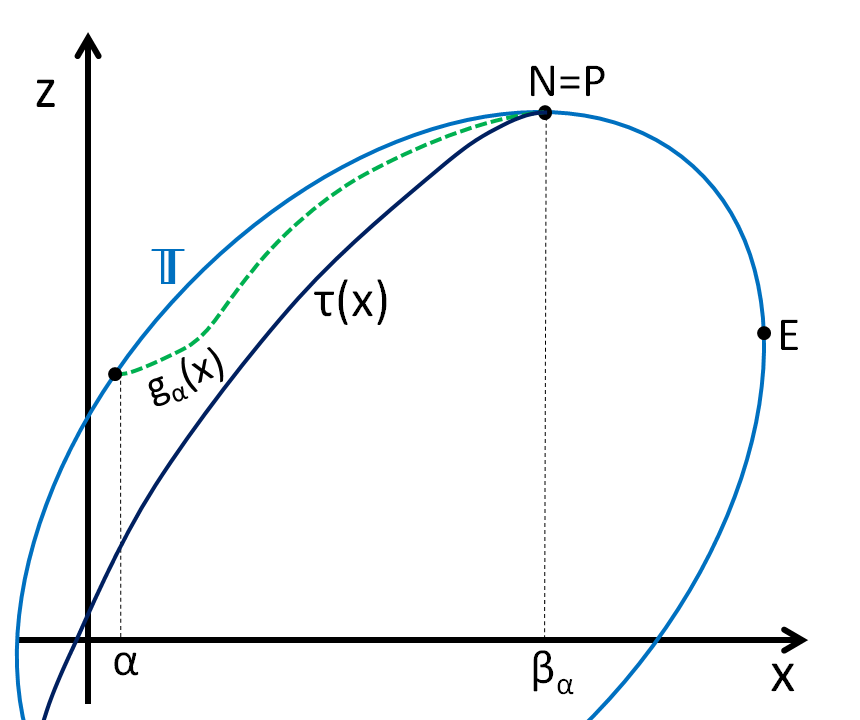} &
\includegraphics[width=3.7cm]{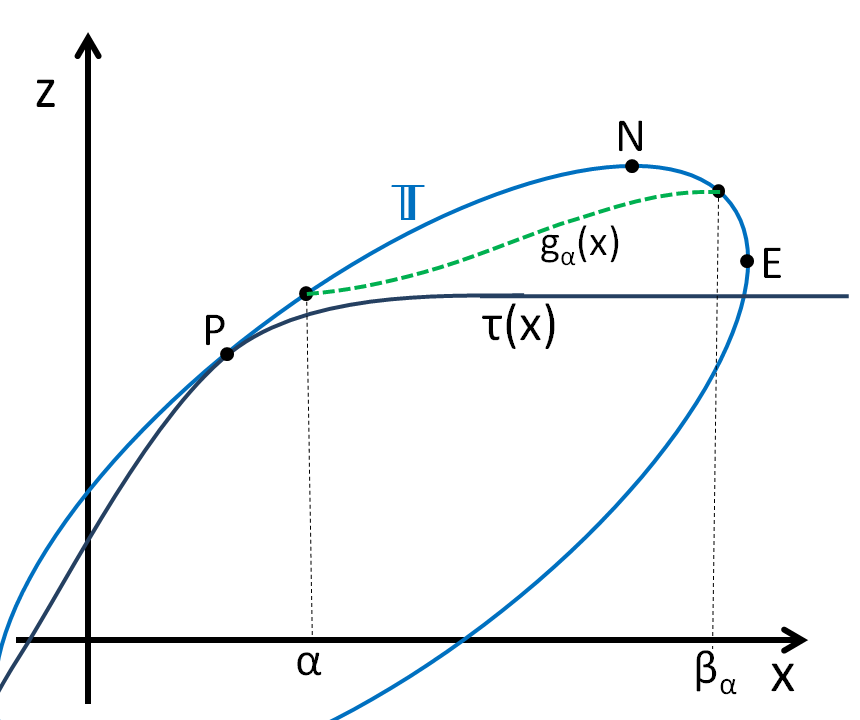} & \includegraphics[width=3.7cm]{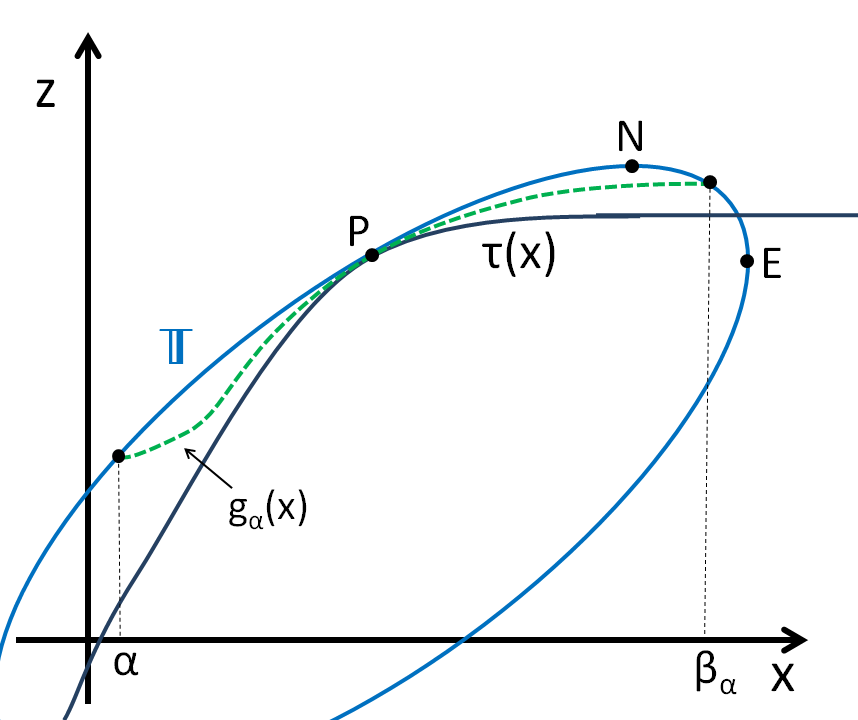} \\
  \pi>1,  \ \alpha>x_P & \pi>1, \ \alpha\leq x_P & \pi<1 & \pi=1  
\end{array}$
\\[2ex]
\figcaption{ $0<p<1, \ \mu<G$}
\end{center}

\bigskip

\noindent A rigorous treatment of the first possibility ($\pi<
1$) is given in the Proposition 
\ref{pro:FIRST} in Section \ref{rigorous}. The other
cases are treated in the Proposition \ref{pro:SECOND}.

\noindent\begin{minipage}{0.7\textwidth}
  \emph{ - Sub-case b): $\mu\geq A$.}\ \ 
  The rate of return in this sub-case is so large, that the value
  function of the problem with transaction costs is infinite,
  independently of the size of the transaction costs $\uld$ and
  $\old$.  A constructive argument is presented in Proposition
  \ref{lem:infinity}.  From the analytic point of view, this phenomenon
  is related to the non-existence of the solution to the free-boundary
  problem;
  an illustration of the reason why is given by
  the picture to the right (Figure 4). 
  In a nutshell, we can find an
  asymptotically linearly increasing curve $T_u(x,K)$ (the notation is
  chosen to fit that of Section 6) such that $\ga$
  stays above it for all $x$. 
  Consequently, it is prevented from
  reaching the other branch of the curve $\bT$ and satisfying the 
  free-boundary condition.
\end{minipage}\hfill
\begin{minipage}{0.26\textwidth}
  \begin{center}
	\includegraphics[width=4.5cm]{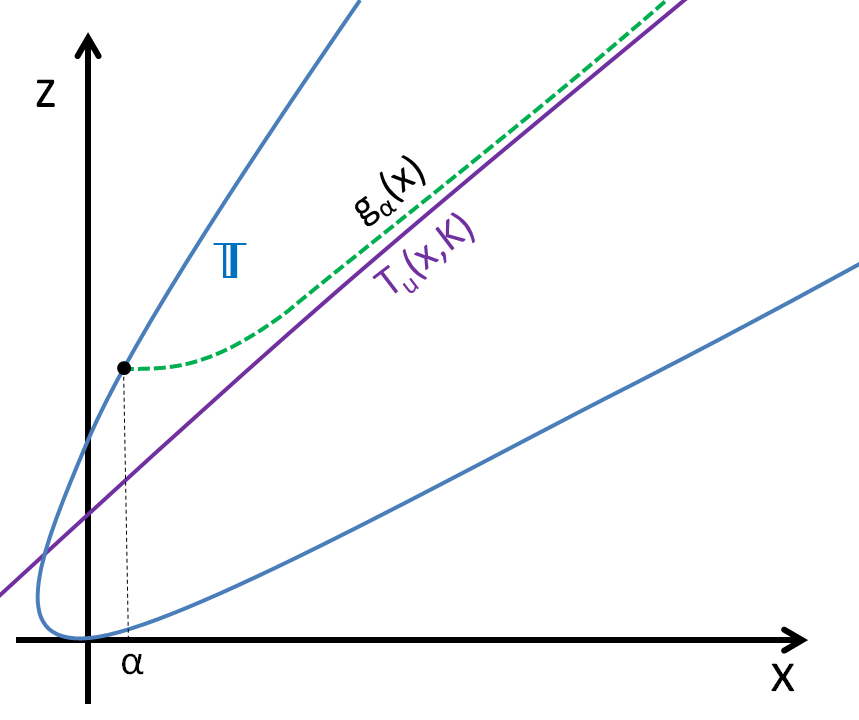}
	\figcaption{}
  \end{center}
\end{minipage}

\medskip

\emph{ - Sub-case c): $G \leq \mu < A$.}\ \ 
This is the most interesting sub-case from the point of view of
well-posedness; whether the value function is finite or not is
determined by the size of the transaction costs.
The curve $\mathbb{T}$ is a hyperbola, 
and has no North pole.
The overall approach is the same as in sub-case a): we construct a
maximal solution $g_{\alpha}$ on an interval of the form
$[\alpha,\beta_{\alpha}]$, and show that the following two
statements hold:
\begin{enumerate}
  \item  The map $\alpha  \mapsto \int_{\alpha }^{\beta _{\alpha}}
  \frac{g_{\alpha}'(x)}{x} dx$ is continuous and strictly decreasing
  on $(0,\infty)$, and 
\item $\lim_{\alpha  \searrow 0} \int_{\alpha }^{\beta _{\alpha}}
  \frac{g_{\alpha }'(x)}{x} dx =\infty$ while
$\lim_{\alpha  \nearrow \infty} \int_{\alpha }^{\beta _{\alpha}}
\frac{g_{\alpha}'(x)}{x} dx = C$,
\end{enumerate}
where an expression for $C=C(\mu,\sigma,\delta,p)$ is given in \eqref{ite:C-expression}
below.
The reader will note two major differences when the statements here
are
compared to the corresponding statements in the sub-case a). The first
one is that
$+\infty$ now plays the role of $x_N$. The second one is that the
range of the integral $\int_{\alpha }^{\beta _{\alpha}}
\frac{g_{\alpha}'(x)}{x} dx $ is not the set of positive numbers
anymore. It is an interval of the form
$(C,\infty)$, which makes the free-boundary problem
solvable only for 
$\log(\tfrac{1+\old}{1-\uld})>C$.

In addition to the fact that we still need to deal with the possible
singularity along the graph $\Gamma_{\alpha}$ of $g_{\alpha}$,
difficulties of a different nature appear in this sub-case.
First of all, due to the unboundedness of the regions separated by a
hyperbola, it is not clear whether the maximal solution started at
$x=\alpha$ will ever hit the curve $\bT$ again. Indeed, this is
certainly a possibility when $\mu\geq A$, as depicted in Figure 4. 
However, we prove by contradiction that this is not the case for $G\leq \mu < A$. 
The second new difficulty has to do with fact that $C$ is finite - a
fact which prevents the existence of a solution to \eqref{equ:HJB-g},
\eqref{equ:integral-cond}.
% middle mu
\begin{center}$
\begin{array}{ccc}
\includegraphics[width=4.7cm]{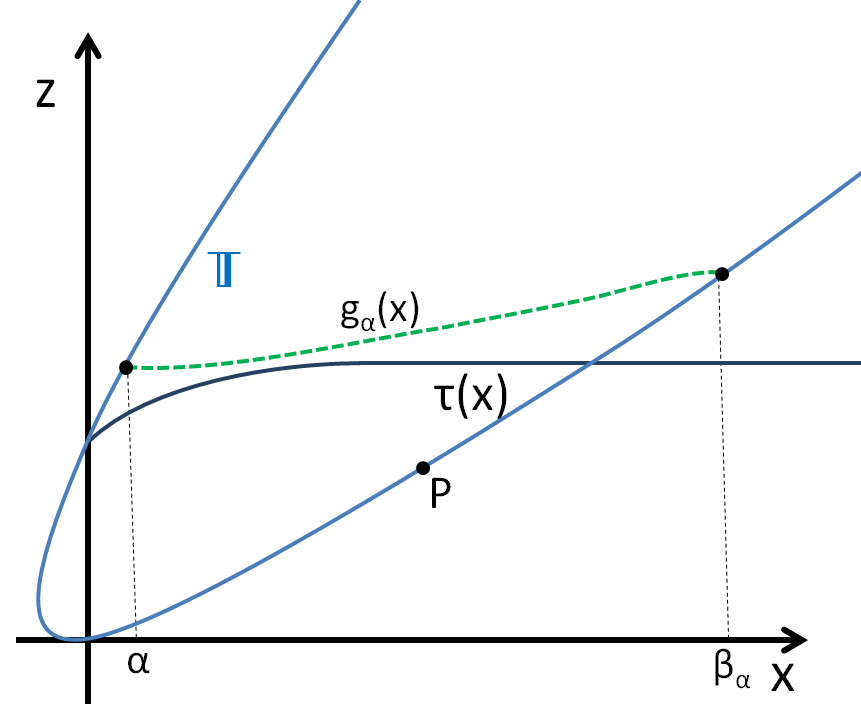} &\qquad
\includegraphics[width=4.7cm]{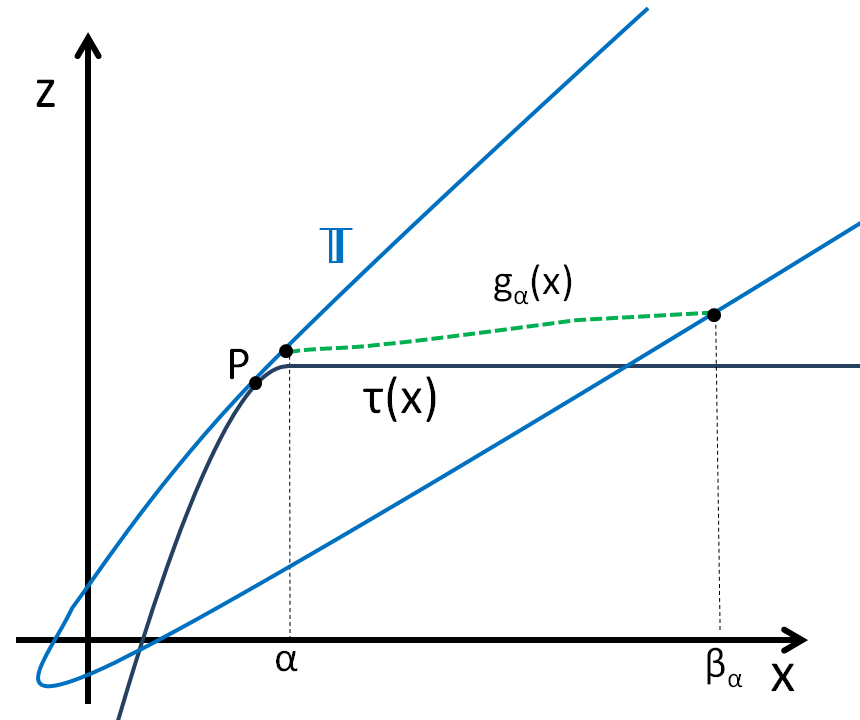} &\qquad
\includegraphics[width=4.7cm]{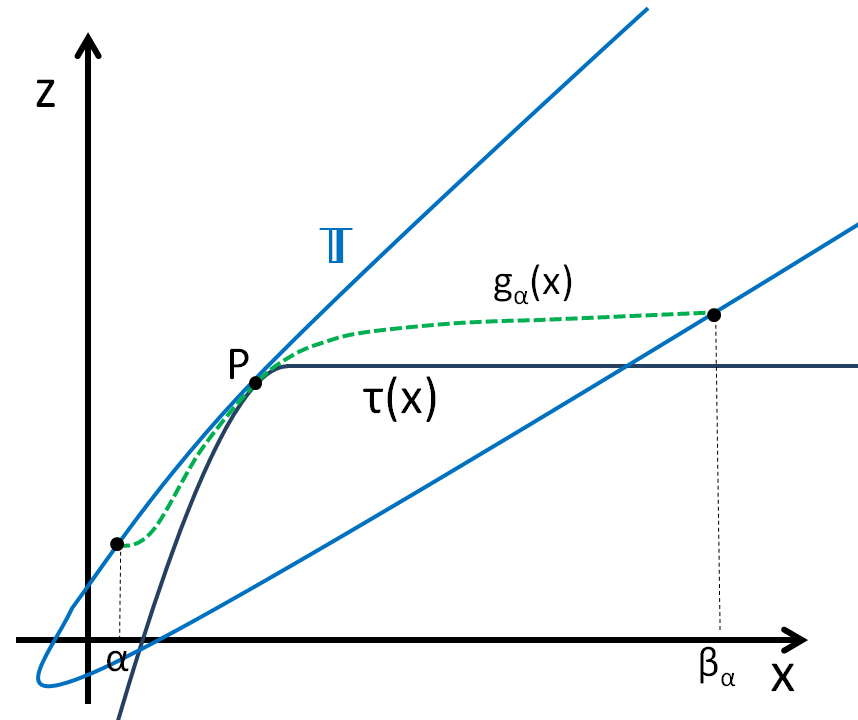}  \\
\pi<1 &\qquad \pi>1,  \ \alpha>x_P &\qquad \pi>1, \ \alpha\leq x_P
\end{array}$
\figcaption{$0<p<1, \ G\leq \mu<A$}
\end{center}

The rigorous treatment of this sub-case
is in Propositions \ref{pro:THIRD} and \ref{prop:wellposed} in Section 6.
Figure 5.~below
illustrates three representative regimes. We note that
the equality $\pi=1$ cannot hold for $\mu\in [A,G)$, as it would
force $A=G$. 

\subsection{High risk aversion $p\leq 0$} 
In this case the problem is
always well posed independently of the values of $\uld$ and $\old$; indeed, the utility
function is bounded from above. The curve $\mathbb{T}$ is
a hyperbola for $p<0$ and a parabola for $p=0$, and it has a
North-pole for any $p\leq 0$.  $\bB$ is a hyperbola for $p<0$, and for
$p=0$, it is a union of two straight lines, one of which is $x=1$.

Compared to the case $0<p<1$, no major new difficulties arise here,
even though one still has to deal with the existence of singularities. 
For this reason we only
present a figure (Figure 6.~below) which illustrates different sub-cases that may
arise. The formal treatment is analogous to that of Section \ref{rigorous}.
\begin{figure}[htb]
	\begin{center}$
	  \begin{array}{cccc}
\includegraphics[width=0.23\textwidth]{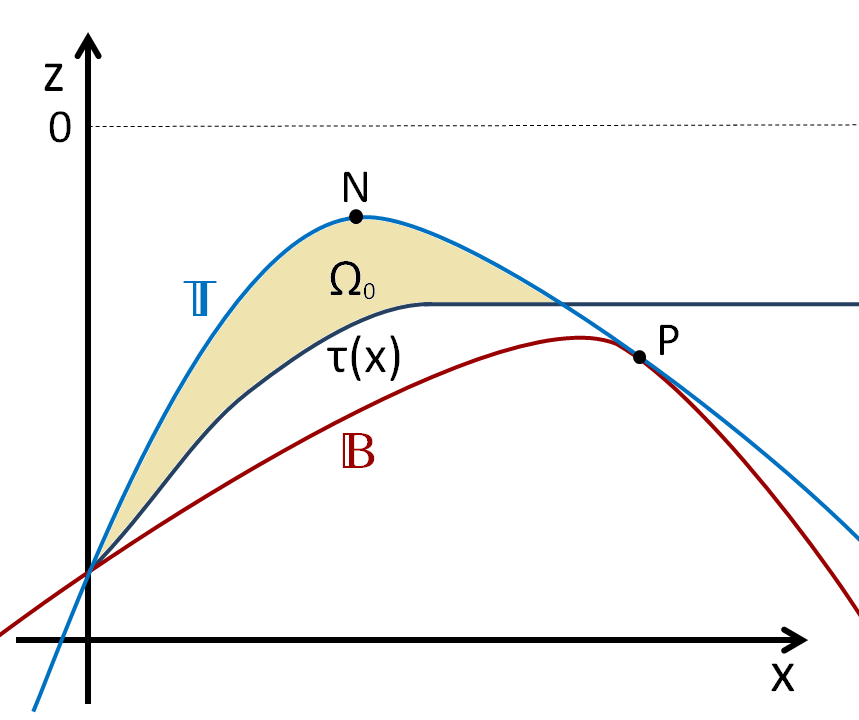} &
\includegraphics[width=0.23\textwidth]{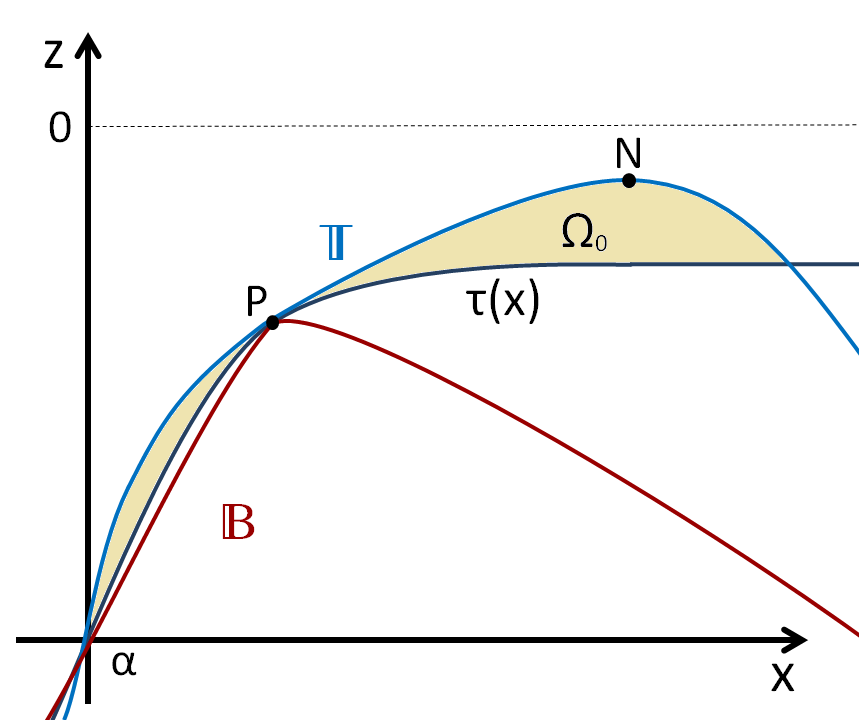} 
&\includegraphics[width=0.23\textwidth]{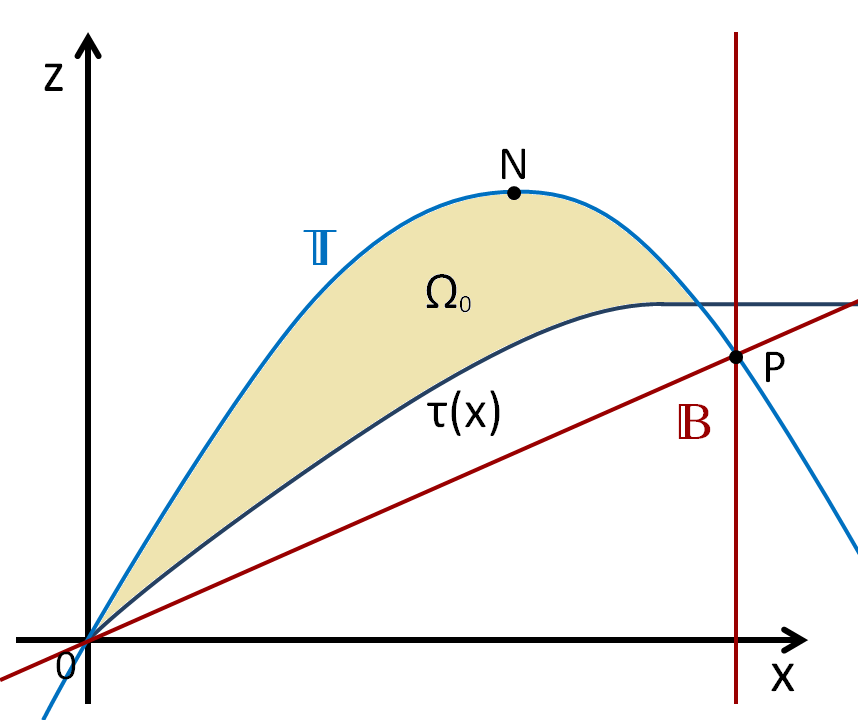} &
\includegraphics[width=0.23\textwidth]{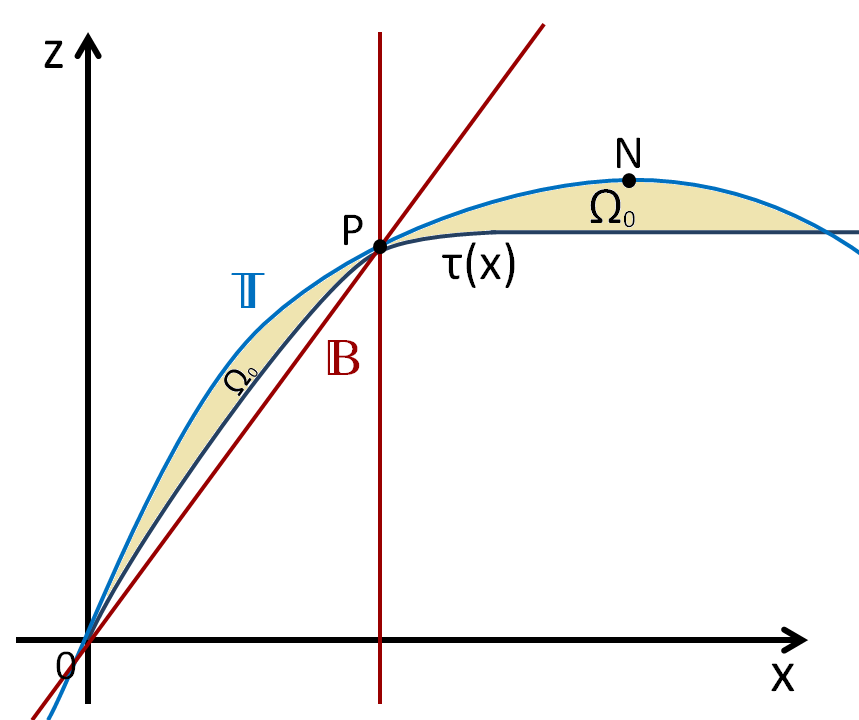} \\
p<0, \ \pi<1 & p<0 \ \pi=1 & 
p=0, \ \pi<1 & p=0, \ \pi>1\\
\end{array}$
\end{center}
\caption{$p\leq 0$}
\end{figure}

\section{An existence proof for the free-boundary problem
  \eqref{equ:HJB-g}, \eqref{equ:integral-cond}}
\label{rigorous}
After a heuristic description of the major steps in the existence
proof and the
associated difficulties, we now proceed to give more rigorous, formal
proofs. 
More precisely, the goal of this section is to
present a proof of the part \eqref{ite:main-3} of Theorem
\ref{thm:main}.

As already mentioned in the previous section, the proofs in
the case $p\leq 0$, are very similar (but less involved) than those in
the case $p\in (0,1)$ so we skip them and refer the reader to the
first author's PhD dissertation \cite{Cho12} for details. We also
do not provide the proof of the part (c) of Theorem \ref{thm:main}, as
it can be obtained easily by an explicit computation.

\medskip

Out first result states that
problem is not well posed for large $\mu$.
  As a consequence, we will be focusing on the case $0<p<1$, $\mu<A$
  in the sequel.
\begin{proposition}[$0<p<1$, $\mu\geq A$.]\label{lem:infinity}
If $0<p<1$ and  $\mu \geq A$, then
$u=\infty$, for all 
$\uld,\old\geq 0$.
\end{proposition}
\begin{proof}
Without loss of generality, we consider the case $\eta_B=0,\eta_S=1$,
and construct a portfolio $(\vp^0,\vp,c)$ as follows:
$$ \varphi^0_t = 0,\  \varphi_t = (t+1)^{-\frac{1-p}{p}}\text{ and }
c_t = \tfrac{1-p}{p} (1-\uld)\,S_t (t+1)^{-\frac{1}{p}},\text{ for
}t\geq 0. $$
One easily checks that it is admissible and that its expected utility
is given by 
\begin{equation}
\nonumber
\begin{split}
	\mathbb{E} \Big[\int_0^{\infty} e^{-\delta t} \tfrac{ c_t^p}{p} dt
\Big]& = \tfrac{(1-p)^p(1-\uld)^p}{p^{1+p}} \mathbb{E}
\Big[\int_0^{\infty} e^{-\delta t} \tfrac{S_t^p}{t+1} dt \Big] 
 =\tfrac{(1-p)^p(1-\uld)^p}{p^{1+p}} \int_0^{\infty}
 e^{pt(\mu-A)}\tfrac{1}{t+1}dt = \infty.\qedhere
 \end{split}
\end{equation}
\end{proof}
\subsection{Maximal inner solutions of $g'=L(\cdot,g)$.}
As explained in the previous section, the main technique we employ in
all of our existence proofs is the construction
of a family of solutions to the equation $g'=L(\cdot,g)$, followed by
the choice of 
the one that satisfies the appropriate integral condition. We,
therefore, take some time here to define the appropriate notion of a
solution to a singular ODE $g'=L(\cdot,g)$:
\begin{definition}
  Let $\sD$ be a convex interval in $(0,\infty)$.
  We say that a
  function  $g:\sD\to\R$ is a
  \define{continuous solution} of the equation $g'=L(\cdot,g)$ if
  \begin{enumerate}[leftmargin=1.8em]
    \item $g$ is continuous on $\sD$, 
	\item $g$ is differentiable at $x$ and $g'(x)=L(x,g(x))$, for all
	  $x\in\Int\sD\setminus\set{x_P}$
  \end{enumerate}
\end{definition}
\noindent We note that any function with a single-point domain $\sD=\set{x}$ is
considered a continuous solution according to the above definition.

Remembering that $p\in (0,1)$ and 
using the notation of the previous section
we remark that the level curves $L=k$ are ellipses or hyperbolas, and, as such, they 
are not graphs in general.
\begin{figure}[h!]
  \noindent\begin{minipage}{0.57\textwidth}
We therefore introduce the \define{upper graph} $T_u(x,k)$ and the
\define{lower graph} $T_d(x,k)$ of the level curve $L=k$ by
\[ T_u(x,k)=\max 
  \sets{z\in\R}{ P(x,z)= k\,  Q(x,z)}\] and \[
  T_d(x,k) =
  \min
  \sets{z\in\R}{ P(x,z) = k\, Q(x,z)},
\] 
for all $x\in\sL_k$, where
\[ \sL_k=\sets{x>0}{ P(x,\cdot) = k\, Q(x,\cdot)\text{ admits a solution.}}\]
Moreover, for convenience, we include the case $k=\infty$, where the
minimal and maximal solutions of $Q(x,\cdot)=0$ (instead of
$L(x,\cdot)=k$) are considered; the domain $\sL_{\infty}$ is also
defined. 
One easily checks that
\begin{displaymath}\begin{split}
	\mathbb{T} &= \sets{(x,z)}{ x\in \sL_0,\,  z=T_u(x,0) \textrm{  or
	  } T_d(x,0)},\text{ and }\\
	\mathbb{B} &= \sets{(x,z)}{ x\in\sL_{\infty},\, z=T_u(x,\infty)
	  \textrm{  or  } T_d(x,\infty)}.
\end{split}
\end{displaymath}
Functions $T_u$ and $T_d$ allow us to define a subclass of solutions
to $g'=L(\cdot,g)$:
\begin{definition} A continuous solution 
  $g:\sD\to\R$ is said to be a \define{maximal inner solution} if
  \begin{enumerate}[leftmargin=1.8em]
	\item $T_d(x,0) \leq g(x) \leq T_u(x,0)$, for all $x\in \sD$, and
	\item $g$ cannot be extended to an interval strictly larger than
	  $\sD$, without violating either (1) or the continuous-solution
	  property. 
  \end{enumerate}
\end{definition}
\end{minipage}
\hfill
\begin{minipage}{0.42\textwidth}
\begin{center}
$\begin{array}{c}
\includegraphics[width=5cm]{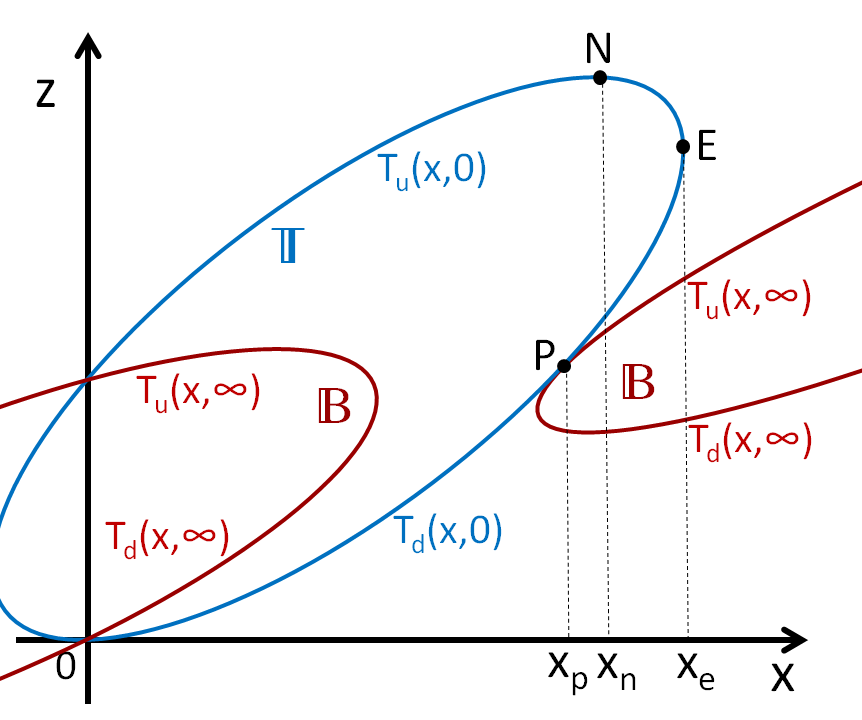}\\
\mu<G \\[4ex]
 \includegraphics[width=5cm]{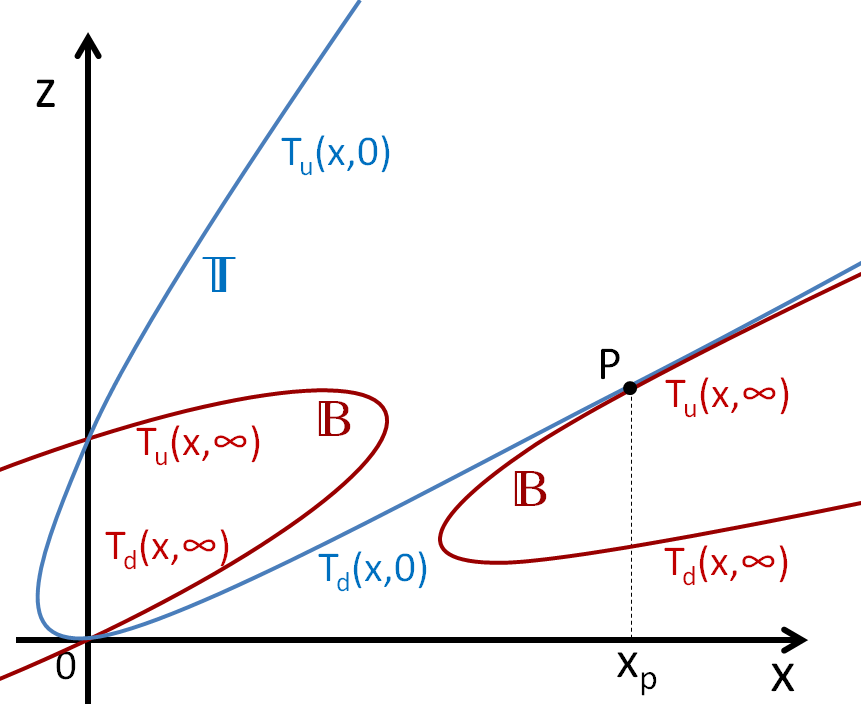} \\
 \mu \geq G\\
\end{array}$
\caption{
}\end{center}
\end{minipage}
\end{figure}

Thanks to the local Lipschitz property of the function $L$ away
from $\bB$, the general theory of ordinary differential equations,
namely the Peano Existence Theorem (see, for
example Theorem I, p.~73 in \cite{Wal98}), states that, starting from any point $(x,z)$
with $z\in [
T_d(\cdot,0),T_u(\cdot,0)]$ and 
$L(x,y)$ well defined, one can construct a maximal inner solution
$g:\sD\to\R$. This solution is necessarily real-analytic away from the curve $\sB$. 

We will be particularly interested in maximal inner solutions started
at the top portion of $\bT$, i.e., at the point $(\alpha, T_u(\alpha,0))$, for
$\alpha\in\sL_0\setminus\set{x_P}$. Not
assuming uniqueness, we pick one
such solution, denote it by $g_{\alpha}$, its domain by $\sD_{\alpha}$,
and the right boundary point of $\sD_{\alpha}$ by $\beta_{\alpha}\in
(0,\infty]$. To avoid the analysis of unnecessary cases, we assume
from the start that $\alpha\in (0,x_N)\setminus\set{x_P}$, so that $T_u(\cdot,0)$ is
strictly increasing in the neighborhood of $\alpha$ and the
singularity $x_P$ is not used as the initial value (a curious reader
can
peak ahead to Proposition \ref{pro:SECOND}, to see how the case
$\alpha=x_P$ can be handled.)

To rule out the possible encounters of a maximal inner solution with $\bB$ away from the
singular point $x_P$, we delve a bit deeper into the geometry of the
right-hand side $L$ of our ODE. We start by a technical lemma which
will help us construct the containment curve $\tau$. Some more
explicit expressions for the upper and lower curves $T_u$ and $T_d$ are
going to be needed:
\begin{equation}
  \label{T_u expression}
  T_{u,d}(x,k)=\frac{1-p}{2\delta}+ \tfrac{b(k)x \pm \sqrt{(b(k)^2-4a(k)c(k))x^2 + 4p(1-p)(k+1)(b(k)-4\delta k) x + 4p^2(1-p)^2(1+k)^2} }{2a(k)}
\end{equation}
where $a(k),b(k),c(k)$ are given by
\begin{equation}\label{abc}
\left \{ \begin{array}{ll} a(k)=2p\delta(1+k), \\  b(k)= (2\delta + p(1-p)(2\mu-\sigma^2))k + 2p(1-p)\mu, \\ c(k)=(1-p)(2\mu+(p^2-1)\sigma^2)k + (1-p)^3 \sigma^2. \end{array} \right.
\end{equation}
The end-points of the domains of $T_u$ and $T_d$, i.e., those $x$ for
which $T_u(x,k)=T_d(x,k)$ are given by 
\[ x_{\pm}(k)=\frac{k+1}{\pm G(1-\frac{p}{1-p} k)
	-\mu + (\frac{A}{1-p}-\mu)k}.\]
We can also check that $
T_{u,d}(x,\infty):=\lim_{k\to\infty} T_{u,d}(x,k)$
are solutions of $Q(x,\cdot)=0$ and that
$x_{\pm}(\infty):= \tfrac{1}{\frac{p}{1-p}(A\mp G)+ (A-\mu)}$ are the
solutions to $T_d(\cdot,\infty)=T_u(\cdot,\infty)$. 
Finally, we note for future reference that
 $0<x_-(\infty)<x_+(\infty)\leq x_P$ holds, and that, for $\mu<G$, 
the $x$-coordinates of the north and east points ($N,E$) are given by 
$x_N = \tfrac{2\mu}{G^2-\mu^2}$ and $x_E
=x_+(0)=\tfrac{1}{G-\mu}.$

\begin{lemma}\label{prop:k0}
 For $0<p<1$ and $\mu<A$,  there exists constant $k_0\in (0,\frac{1-p}{p})$ such that
  \begin{enumerate}[itemindent=-3ex]
\item $x_+(k_0)\geq x_P$ and $(0,x_+(k_0)]\subseteq \sL_{k_0}$.
\item $\tfrac{\partial}{\partial x} T_u(x,k_0)<k_0$ for $x\in (0,x_+(k_0))$.
\end{enumerate}
\end{lemma}
\begin{proof} 
(1) A direct calculation shows that, for $k=\tfrac{1-p}{p}$, we have
  $x_+(k)=x_P>0$, as well as 
\begin{equation}\nonumber
 \begin{split}
   b^2(k)- 4 a(k) c(k) &= 4 (1-p)^2 (A-\mu)^2>0,\text{ and }
   k- \tfrac{d}{dx}T_u(0,k) = \tfrac{1}{\delta} (1-p)(A-\mu)>0.
 \end{split}
\end{equation}
By continuity, we can find $k_0 \in (0,\frac{1-p}{p})$ such that 
$$b(k_0)^2- 4a(k_0)c(k_0)>0, \ k_0>\tfrac{\partial}{\partial x}
T_u(0,k_0) \text{ and } x_+(k_0)>0.$$
We can check that $x_+(k_0)\leq x_-(k_0)$, which, in turn, implies
that $\sL_{k_0} = (0,x_+(k_0)]\cup [x_-(k_0),\infty)$. Since
$x_+(\frac{1-p}{p}) =x_P$ and $ \tfrac{d}{dk}x_+(\frac{1-p}{p}) \leq 0$, we conclude that $x_+(k_0)\geq
x_P$\\
(2) The result follows from $\tfrac{\partial^2}{\partial x^2} T_u(x,k_0)<0$ and
$\tfrac{\partial}{\partial x} T_u(0,k_0)< k_0$.
\end{proof}

With the constant $k_0$ as in Lemma \ref{prop:k0} above fixed,
we
define the \define{containment curve} $\tau : [0,\infty) \to \bR$ and
a \define{containment region} $\Omega_0 \in \bR^2$ by
\[ \tau(x)= \max_{x'\in[0,x]}T_u(x' \wedge x_+(k_0),k_0),\text{ and }
  \Omega_0 = \sets{(x,z)\in (0,\infty)\times \R}{T_d(x,0) \vee \tau(x)\leq z \leq T_u(x,0)
}.\]
The significance of these objects is made clearer in the following
proposition. The reader is invited to consult Figure 2 for
an illustration.
\begin{proposition}\label{Omega} For $0<p<1$ and $\mu<A$, 
  the following statements hold:
  \begin{enumerate}[itemindent=-3ex]
  \item If $\pi\geq 1$, then $\Omega_0 \cap \bB = \set{P}$, $\tau(x_P)=T_u(x_P,0)$, and
	$x_P\leq x_N$. For $\pi<1$, $\Omega_0\cap \bB = \emptyset$.
\item $\Omega_0$ is simply connected. It is bounded if and only 
  if $\mu<G$.
\item $(\Omega_0 \setminus \set{P}) \cap \{L(x,z)=k\} = (\Omega_0
  \setminus \set{P}) \cap\set{ z=T_u(x,k)\text{ or }z=T_d(x,k)}$.
\item $(\Omega_0 \setminus \set{P}) \cap \{ L(x,z)>k \} = (\Omega_0 \setminus \set{P}) \cap \{ T_d(x,k)<z<T_u(x,k)\}$.
\item $0\leq L(x,z)\leq k_0$ for $(x,z)\in \Omega_0 \setminus \{P\}$.
\item $\tau \in C^1([0,\infty))$. For $x>0$ such that
  $T_d(x,0)<\tau(x)<T_u(x,0)$, we have $\tau'(x) < L(x,\tau(x))$. 
\end{enumerate}
\end{proposition}
\begin{proof}
(1) It is easily checked that, when all sets are seen as subsets of
$(0,\infty)\times \R$, that 
\begin{displaymath}\begin{split}
 \Omega_0\cap \bB &=\Big\{T_d(x,0)\vee \tau(x) \leq z \leq
   T_u(x,0)\Big\}
 \bigcap  \Big(\Big\{x\leq x_-(\infty),\  z=T_d(x,\infty)
  \textrm{  or  }T_u(x,\infty)\Big\}\cup \{P\}\Big).
\end{split}
\end{displaymath}
Hence,  it will be enough to show that following two claims hold:

\smallskip

\noindent {\bf Claim 1:} \emph{For $0<x\leq x_{-}(\infty)$, we have
  $\tau(x)>T_u(x,\infty)$.}\ \ 
Since $T_u(\cdot,\tfrac{1-p}{p})$ is a straight line on $[0,x_P]$ and
$T_u(\cdot,\infty)$ is concave on $[0,x_-(\infty)]$, the easy-to-check
facts that
$$T_u(0,\tfrac{1-p}{p})=T_u(0,\infty)\text{ and } \tfrac{\partial}{\partial x}T_u(0,\tfrac{1-p}{p}) >T_u'(0,\infty) $$
imply that $T_u(\cdot,\tfrac{1-p}{p})>T_u(\cdot,\infty)$, on  $(0,x_-(\infty)]$.
Similarly, $T_u(\cdot,k_0)>T_u(\cdot,\tfrac{1-p}{p})$ on $(0,x_-(\infty)]$.

\smallskip

\noindent{\bf Claim 2:} \emph{$\tau(x_P) \geq z_P$, with equality if
  and only if $\pi \geq 1$.}\ \
Several sub-cases are considered:
\begin{enumerate}[leftmargin=1.8em]
  \item[i)]
 $\pi\geq 1$ : Then, $\tfrac{\delta}{p}>\tfrac{(1-p)\sigma^2}{2}$ and $z_P=T_u(x_P,k_0)$. Since $T_u(\cdot,k_0)$ is strictly concave and 
$\tfrac{\partial}{\partial x}T_u(x_P,k_0)=\tfrac{2(1-p)^2 \sigma^2
  (\pi-1)}{2\delta-p(1-p)\sigma^2} \geq 0$, the map
$T_u(\cdot,k_0)$ is strictly increasing on $[0,x_P]$. Thus, $\tau(x_P)=T_u(x_P,k_0)=z_P$.
\item[ii)]
$\pi <1$, $\tfrac{\delta}{p}>\tfrac{(1-p)\sigma^2}{2}$ :
$\tfrac{\partial}{\partial x}T_u(x_P,k_0)<0$ implies that
$\tau(x_P)>T_u(x_P,k_0)=z_P$.
\item[iii)]  $\pi <1$, $\tfrac{\delta}{p}<\tfrac{(1-p)\sigma^2}{2}$ : $\tau(x_P)\geq T_u(x_P,k_0)>T_d(x_P,k_0)=z_P$.
\item[iv)] $\pi <1$, $\tfrac{\delta}{p}=\tfrac{(1-p)\sigma^2}{2}$ : $\lim_{x\nearrow x_P} \tfrac{\partial}{\partial x}T_u(x,k_0) = -\infty$ implies $\tau(x_P)>T_u(x_P,k_0)=z_P$.
\end{enumerate}

\medskip

\noindent (2) For the simple connectedness of $\Omega_0$, it is enough
to show that $\{x>0 : T_d(x,0)\vee \tau(x) \leq T_u(x,0)\}$ is an
interval. Given that $T_d(x,0) \leq T_u(x,0)$, for all $x$, 
it is enough to show that $\{x>0 : \tau(x) \leq T_u(x,0)\}$ is an
interval. With $x_m\in \argmax_{x\in
  [0,x_+(k_0)]}T_u(x,k_0)$, 
similarly as in the proof of Claim 1, we observe that for $x\in [0,x_m]$,
$T_u(x,k_0)\leq T_u(x,0)\text{ and that }  T_u(\cdot,0) \textrm{
  strictly increases.}$
Therefore,  $(0,x_m]\subset \{x>0 : \tau(x) \leq T_u(x,0)\}$. 
Since
$T_u(\cdot,0)$ is strictly concave and $\tau$ is constant after
$x_m$, we have $T_u(\cdot,0)< \tau$, to the right of the right-most
point at which 
$T_u(\cdot, 0)$ equals $\tau$.

Since $T_u(\cdot,0)$ is strictly concave and $\tau$ is constant after $x_m$, boundedness of $\Omega_0$ is equivalent to the boundedness of the
domain $\sL_0$, of $T_u(\cdot,0)$ and $T_d(\cdot,0)$. The set $\sL_0$
is, in turn, bounded, if and only if $\mu<G$.

  \medskip

  \noindent (3) The statement follows from definitions of $T_u$ and $T_d$, and the fact that $P$ is the
unique singular point.

\medskip

\noindent (4) We only need to observe that
$\big(\Omega_0 \setminus\{P\}\big) \subset \{P(x,z)\geq0\} \cap \{Q(x,z)>0\}$ from (1).

\medskip

\noindent (5) The result follows from (3),(4) and the definition of $\tau$.

\medskip

\noindent (6) $C^1$-smoothness of $\tau$ follows easily from the construction. 
With $x_m$ as in (2) above, we have 
$$\tau(x)= \left\{\begin{array}{ll} T_u(x,k_0), & \textrm{for  } x\in
	[0,x_m],\\
T_u(x_m,k_0), &\textrm{for  } x\in [x_m,\infty). \end{array} \right.$$
For $x\in [0,x_m)$, by Lemma \ref{prop:k0}, we have
$$\tau'(x)= \tfrac{\partial}{\partial x} T_u(x,k_0)<k_0=L(x,T_u(x,k_0))=L(x,\tau(x)).$$
For $x\in [x_m,\infty)$, $\tau'(x)=0$, but the condition
$T_d(x,0)<\tau(x)<T_u(x,0)$ implies that
$$(x,\tau(x))\in \{Q(x,z)>0, \ P(x,z)>0\}\subset \{L(x,z)>0\}.\qedhere
$$
\end{proof}
With the result of Proposition \ref{Omega} in hand, we can say 
more about the shape of the function $\ga$ and its domain
$\sD_{\alpha}$. In particular, we show that the graph
$\Gamma_{\alpha}$ stays at a positive distance from any point of
$\mathbb{B}$, except, maybe, $P$. Remember that $x_N=\infty$ if $\mu \geq G$.
\begin{proposition}\label{Gamma}
  For $\alpha\in (0,x_N)$ such that $(\alpha,T_u(\alpha,0))\neq P$, we have
  \begin{enumerate}[leftmargin=1.8em]
	\item $(\Gamma_{\alpha}\setminus\{P\})\cap \{(x,T_u(x,0)):\alpha<x\leq x_N\}=\emptyset$.
	\item $\Gamma_{\alpha} \subseteq \Omega_0$.
	\item $\sD_{\alpha}$ is a closed
  interval of the form $[\alpha,\beta_{\alpha}]$, for some
  $\beta_{\alpha}\in (\alpha,\infty]$. If $\beta_{\alpha}<\infty$ and $(\beta_{\alpha},\ga(\beta_{\alpha}))\neq P$, then $\ga'(\beta_{\alpha})=0$.
   \end{enumerate}
 \end{proposition}
 \begin{proof} 

(1) Suppose that it is not true. Then, there exists $x_0\in(\alpha,x_N]$ such that $\ga(x_0)=T_u(x_0,0)$ and $(x_0,\ga(x_0))\neq P$. Since $T_u(x,0)\geq \ga(x)$ for $x<x_0$ close enough, we have $\tfrac{\partial}{\partial x}T_u(x_0,0) \leq \ga(x_0)$. Combine this with $\ga'(x_0)=0$, we deduce that $x_0=x_N$ and $\mu<G$. So, $\ga(x_N)=T_u(x_N,0)$ and $\ga'(x_N)=0=\frac{\partial}{\partial x}T_u(x_N,0)$. Using this, we can calculate $\ga''(x_N)=\frac{d}{dx}L(x,\ga(x))\vert_{x=x_N}=0>\frac{\partial^2}{\partial x^2}T_u(x_N,0)$, which contradict to the fact that $T_u(x)\geq \ga(x)$ for $x<x_N$ close enough.

   (2) Noting that
$(\alpha,\ga(\alpha))\in\Omega_0$, 
we assume that there exists a point $x\in \sD_{\alpha}$ with
$(x,\ga(x))\not \in\Omega_0$. With $x_0$ denoting the infimum of all such
points, we observe immediately that $x_0>\alpha$ and $x_0<x_E\in
(0,\infty]$. 
If, additionally, $x_0\ne x_P$, we can use the continuity of $\ga$ to conclude that
$\ga(x_0)=\tau(x_0)$ and $\ga'(x_0)\leq \tau'(x_0)$. By using (1), 
we can exclude the case $\ga(x_0)\in \{T_u(x_0,0),T_d(x_0,0)\}$, otherwise, $D_{\alpha}$ should be $[\alpha,x_0]$, which contradicts to the choice of $x_0$.
Then,
since $T_d(x_0,0) < \tau(x_0) < T_u(x_0,0)$, we reach a contradiction
with part (6) of Lemma \ref{Omega}:
\[ L(x_0,\ga(x_0)) = \ga'(x_0) \leq \tau'(x_0) <
  L(x_0,\tau(x_0))=L(x_0, \ga(x_0)).\]

In the case $x_0=x_P$, we have $\ga(x_0)=\tau(x_0)$ and, by the
definition of the point $x_0$ and the domain $\Omega_0$, there exists
$x'>x_0$, $x'\in \sD_{\alpha}$ such that $\ga(x')<\tau(x')$.
Consequently, we have $\ga'(x'')< \tau'(x'')$ for some
$x_0<x''<x'$, and we observe that $L(x'',\ga(x''))\geq L(x'',\tau(x''))$, because $L(x,z)$ is a decreasing function of $z$ near $P$. Now we reach a contradiction as in the case $x_0\neq x_P$:
\[ L(x'',\ga(x'')) = \ga'(x'') <\tau'(x'') \leq
  L(x'',\tau(x''))\leq L(x'', \ga(x'')),\]
where $\tau'(x'') \leq  L(x'',\tau(x''))$ can be showed by part (6) of Lemma \ref{Omega}.
\medskip

\noindent (3) We first observe that the initial value $(\alpha,\ga(\alpha))$ lies on the graph of the function $T_u(\cdot,0)$ which is strictly increasing in the neighborhood of $\alpha$. Therefore, any extension of $\ga$ to the left of $\alpha$ would cross $\bT$ and exit the set $\set{ z\leq  T_u(x,0)}$. So, left end-point of the domain $\sD_{\alpha}$ should be $\alpha$.

To deal with the right end-point of $\sD_{\alpha}$, we note that noe of the following must occur: 1) $\ga$ explodes, 2) $\bB$ cossed, or 3) $\bT$ is hit. The first possibility is easily ruled out by the observation that no explosion can happen without $\ga$ crossing the curve $\bT$, first. The second possibility is severely limited by (2) above; indeed, with part (1) of Proposition \ref{Omega}, $\Gamma_{\alpha}\cap \bB \subset \set{P}$. It is clear now that, in the right end-point limit $\beta_{\alpha}$, the function
  $\ga$ hits $\bT$, provided $\beta_{\alpha}<\infty$. 
  For $\beta_{\alpha}<\infty$ and $(\beta_{\alpha},\ga(\beta_{\alpha})) \neq P$, $\lim_{x\nearrow
	\beta_{\alpha}} \ga(x)$ clearly exists, and, so,
  $\beta_{\alpha}\in \sD_{\alpha}$. Furthermore, $\ga'(\ba)=0$ since $(\ba,\ga(\ba))\in \bT \setminus\{P\}$.
  In case $\beta_{\alpha}<\infty$ and $(\beta_{\alpha},\ga(\beta_{\alpha}))= P$, we also conclude that $\beta_{\alpha}\in \sD_{\alpha}$, by observing that
  $\tau(x) \leq \ga(x) \leq T_u(x,0)$ for $x<x_P$ and
  $\lim_{x\nearrow x_P} T_u(x,0) = \lim_{x\nearrow x_P} \tau(x)$. 
\end{proof}

\subsection{The sub-case $\mu<G$}
We focus on the case $\mu<G$ in this subsection. The curve $\bT$ is
now an ellipse and it admits a north pole with the $x$-coordinate
$x_N<\infty$. By part (1) of Proposition \ref{Omega} and (2) of Proposition \ref{Gamma}, we have the
following dichotomy, valid for all $\alpha\in (0,x_N)$.
\begin{enumerate}[leftmargin=1.8em]
	\item For $\pi<1$, $P\not\in \Gamma_{\alpha}$, and
	 \item For $\pi\geq 1$, $P\in \Gamma_{\alpha}$ if and only if
	   $x_P\in \sD_{\alpha}$. 
\end{enumerate}
We start with the first possibility which avoids the singularity $P$
altogether.
\begin{proposition}[$0<p<1$, $\mu<G$, $\pi<1$] 
  \label{pro:FIRST} 
  Suppose that $0<p<1$, $\mu<G$ and $\pi<1$. Then,
  $\beta_{\alpha}\in (\alpha, x_E]$ and  $\ga$ is
  of class $C^{\infty}$ on $\Int\sD_{\alpha}$, for all $\alpha\in
  (0,x_N)$. Moreover, the function
  $G(\alpha) = \int_{\alpha}^{\beta_{\alpha}} \tfrac{\ga'(x)}{x}\,
	dx$ 
  has the following properties:
  \begin{equation}\label{equ:props1}
   \begin{split}
	 \text{$G$ is continuous on $(0,x_N)$},
	\textstyle\lim_{\alpha\searrow 0} G(\alpha) = +\infty,
	\text{ and }
	\lim_{\alpha\nearrow x_N} G(\alpha) = 0.
   \end{split}
  \end{equation}
  In particular, $\ga$ solves the free-boundary problem 
  \eqref{equ:HJB-g}, \eqref{equ:integral-cond}, for some $\alpha\in
  (0,x_N)$. 
\end{proposition}
\begin{proof}
   By part (2) of Proposition \ref{Omega}, $\beta_{\alpha}$ is bounded and in $(\alpha,x_E]$. $\ga'(\alpha)=\ga'(\beta_{\alpha})=0$ is a consequence of part (3) of Proposition \ref{Gamma}.
Since $P\not\in\Gamma_{\alpha}$, smoothness of $\ga$ follows from the general theory (Peano's
  theorem).
  Moreover, the existence of the initial value 
  $\alpha$, with the desired properties, is a direct consequence
  of the listed properties of $G$, by way of the intermediate value theorem.
  We, therefore, focus on \eqref{equ:props1} in the
  remainder of the proof, which is broken into several claims. The
  proof of each claim is placed directly after the corresponding
  statement.
 
  \medskip

  \noindent  {\bf Claim 1:} {\em If $\ga(\ba)=T_u(\ba,0)$, then $\ba>x_N$.}
 This follows from part (1) of Proposition \ref{Gamma}. 
\medskip
  
\noindent  {\bf Claim 2}: {\em The map $\alpha\mapsto \beta_{\alpha}$ is
	continuous}.\  \ For this, we use the implicit-function theorem and
  the continuity of $\ga$ with respect to the initial data
  (see, e.g., Theorem VI., p~145 in \cite{Wal98}). To be able to use the implicit-function theorem,
  it will be enough to observe that, $\ga(\cdot)$ is not tangent
  to $T_u(\cdot,0)$ (or $T_d(\cdot,0)$) at $x=\ba$, which is a consequence of $\ga'(\beta_{\alpha})=0$ and Claim 1. above.

  \medskip

  \noindent {\bf Claim 3}: {\em The map $\alpha\mapsto G(\alpha)$ is
	continuous.}
 It suffices to use the dominated convergence theorem. Its conditions
 are met, since 
 $\ga'(x)\in [0,k_0]$ (by Proposition \ref{Omega}, part
 (5)).

\medskip

\noindent {\bf Claim 4:} {\em  $\lim_{\alpha \searrow 0} G(\alpha)=\infty$.}\ \ 
The joint continuity of
$\tfrac{\partial}{\partial x} T_u(x,k)$ at $(0,0)$ and the fact that
$\tfrac{\partial}{\partial x} T_u(0,0)=\tfrac{(1-p)\mu}{\delta}>0$,
imply that
there exists $\epsilon>0$ such that
\[\tfrac{\partial}{\partial x} T_u(x,\epsilon)>2\epsilon \textrm{  for
  } x\in [0,\epsilon].\] We define
$l(\alpha)=\alpha+\tfrac{T_u(\alpha,0)-T_u(\alpha,\epsilon)}{\epsilon}$
and remind the reader that $T_u(0,k)=\tfrac{1-p}{\delta}$ for each $k$,
so that $\lim_{\alpha \searrow 0} l(\alpha)=0$. Hence,  we can pick
$\alpha_{\eps}>0$ such that
$l(\alpha)<\eps$, for $\alpha<\alpha_{\eps}$.

For any given $\alpha\in (0,\alpha_{\epsilon})$, if it so happens that
$\ga(x)>T_u(x,\epsilon)$ for  $x\in [\alpha,l(\alpha)]$, then
Proposition \ref{Omega}, part (4), implies that
$\ga'(x)<\epsilon$ on $[\alpha,l(\alpha)]$.
Therefore,
\begin{displaymath}\begin{split}
	0&<\ga(l(\alpha))-T_u(l(\alpha),\epsilon) =
	\int_{\alpha}^{l(\alpha)}\Big( \ga'(x)-\tfrac{\partial}{\partial
	  x}T_u(x,\epsilon)\Big )dx + T_u(\alpha,0)-T_u(\alpha,\epsilon)\\
	&\leq \int_{\alpha}^{l(\alpha)}( \epsilon-2\epsilon)dx +
	T_u(\alpha,0)-T_u(\alpha,\epsilon) = 0, \end{split}
\end{displaymath} which is contradiction. 
We conclude that
 $\ga$ intersects
$T_u(\cdot,\epsilon)$ on $[\alpha, l(\alpha)]$, for each
$\alpha\in(0,\alpha_{\eps})$. 

Using the fact that
$\frac{\partial}{\partial x}T_u(x,\epsilon)>L(x,T_u(x,\epsilon))$ on $[0,\epsilon]$,
we conclude that $\ga(x)<T_u(x,\epsilon)$ on $[l(\alpha),\epsilon]$. By Proposition
\ref{Omega}, part (4) and the fact that $\tau(x)>T_d(x,\epsilon)$ for small
$x$, we have that  $\ga'(x)\geq \epsilon$ on $[l(\alpha),\epsilon]$.
Therefore, $$\liminf_{\alpha \searrow 0} G(\alpha)
\geq \liminf_{\alpha \searrow 0 }
\int_{l(\alpha)}^{\epsilon} \frac{\epsilon}{x} dx = \liminf_{\alpha
  \searrow 0 }  \epsilon \ln{(\frac{\epsilon}{l(\alpha)})} = \infty.
$$ 

\medskip

\noindent {\bf Claim 5:} {\em $\lim_{\alpha\nearrow x_N} G(\alpha)=0$.}\ \ 
We start with the inequality $T_u(\alpha,0)=\ga(\alpha)<\ga(\ba) \leq T_u(x_N,0)$, 
which implies that 
$\lim_{\alpha
  \nearrow x_N}\ba = x_N$. Thus, by Proposition \ref{Omega}, part (5),
we have
\[\limsup_{\alpha \nearrow x_N} G(\alpha)
\leq \limsup_{\alpha \nearrow
  x_N}\frac{(\ba-\alpha)k_0}{\alpha} = 0. \qedhere\]
\end{proof}

Before we move on to the case $\pi\geq 1$, we need a few facts about a  
specific, singular, ODE.
\begin{lemma}
  \label{lem:f}
  Given $\eps>0$, consider the ODE
  \begin{equation}
	\label{equ:f}
	\begin{split}
	  h'(y) = - \tfrac{h(y)}{A(y)y^2}+B(y),
   \end{split}
  \end{equation}
  where $A,B:[-\eps,\eps]$ are continuous functions, with $A(0)> 0$.
   Then,
  the following statements
  hold:
  \begin{enumerate}[leftmargin=1.8em]
	\item 
	  There is a single solution $h_+$ of \eqref{equ:f} on $(0,\eps]$ with
	  $\lim_{y\searrow 0} h_+(0)=0$.
	\item No solutions $h_+$ exist with $\lim_{y\searrow 0}
	  h_+(y)=c\in\R\setminus\set{0}$. 
	\item For any solution $h_-$ on $[-\eps,0)$, we have
	  $\lim_{y\nearrow 0} h_-(y)=0$.
	\item Any function $h:[-\eps,\eps] \to \R$ of the form \[h(y) =
		h_+(y) \ind{y>0} + h_-(y)\ind{y<0},\] 
	  where $h_+$ is as in (1) above, and $h_-$ is {\em any} function
	  as in (3) above, is a $C^1$-solution to
	  \eqref{equ:f}.
  \end{enumerate}
\end{lemma}
\begin{proof}
Elementary transformations can be used to show that
for any solution $h$ of \eqref{equ:f} defined on $[-\eps,\eps]\setminus\set{0}$,
there exist constants $c_1$ and $c_2$ such that
\[ h(y) = \begin{cases}
	e^{D(y)} \big(c_1 -
	  \int_{y}^{\eps} B(t)e^{-D(t)} dt \big), & y\in(0,\eps], \\
e^{D(y)} \big(c_2 + \int_{-\eps}^{y} B(t)e^{-D(t)} dt \big),
& y\in[-\eps,0), \end{cases}\text{ where }
  D(y) = \begin{cases}
	\int_y^{\eps} \tfrac{1}{A(t)} \frac{ dt}{t^2}, & y\in(0,\eps],\\
	\int_{-\eps}^y -\tfrac{1}{A(t)} \frac{dt}{t^2}, & y\in[-\eps,0)
\end{cases}\]
We first note that $\lim_{y\searrow 0}D(y)=\infty$. 
So, to satisfy the condition $\lim_{y\searrow 0}h_+(y)=c\in \R$, the only possibility is
$c_1=\int_0^{\epsilon} B(t)e^{-D(t)}dt$. Then, the L'Hospital's rule implies that
\[ \lim_{y\searrow 0}h_+(y)=\lim_{y\searrow 0} e^{D(y)}\int_{0}^{y} e^{-D(t)} B(t)\, dt = 
  \lim_{y\searrow 0} \tfrac{B(y)}{  \tfrac{1}{A(y)y^2} } =0,\]
and we immediately conclude (1) and (2).

As far as (3) is concerned, since $\lim_{y\nearrow 0}e^{D(y)}c_2= 0$, 
for any $c_2$, the limiting behavior is independent of
$c_2$. Moreover, another use of the L'Hospital's rule implies that
$h(y)\to 0$, as $y\nearrow 0$, for each $c_2\in\R$. 

It remains to show (4), and, for this, we start by computing the
derivative at $0$ of $h$. Like above, we use the L'Hospital rule and
the explicit expression for $h$:
\[ \lim_{y\searrow 0} \tfrac{h(y)-h(0)}{y} = \lim_{y\searrow 0}
  \tfrac{B(y)}{ 1+ \tfrac{1}{A(y)y} } = 0.\]
 Similarly, $\lim_{y\nearrow 0} \tfrac{h(y)-h(0)}{y}=0$, and, so $h'(0)=0$.
 To establish that $\lim_{y\to 0} h'(y)=h'(0)=0$, 
we first use the L'Hospital rule to compute
$\lim_{y \to 0} \frac{h(y)}{y^2} = \tfrac{B(0)}{1/A(0)}$, 
so that, using the equation \eqref{equ:f} for $h$, we can immediately deduce that
$\lim_{y\to 0} h'(y) = 0$. 
\end{proof}
\begin{proposition}[$0<p<1$, $\mu<G$, $\pi\geq 1$]
  \label{pro:SECOND} 
  Suppose that $0<p<1$, $\mu<G$.
  \begin{enumerate}[leftmargin=1.8em]
  \item If $\pi>1$ and 
  \begin{enumerate}[leftmargin=1.4em]
	\item[a)] $\alpha\in(x_P,x_N)$. Then $\beta_{\alpha}\in (\alpha,x_E]$, $\ga$ is
	  of class $C^{\infty}$  and $P\not\in\Gamma_{\alpha}$.
	  \item[b)] $\alpha=x_P$. Then the limits
		\[ 	\textstyle\beta_{x_P}=\lim_{\alpha\searrow x_P} \beta_{\alpha},
		  \text{ and } g_{x_P}(x) = \lim_{\alpha\searrow x_P} \ga(x),\
		  x\in (x_P, \beta_{x_P}],\]
		exist and define a continuous solution to \eqref{equ:HJB-g} with
		the domain
		$[x_P,\beta_{x_P}]$.
	  \item[c)] $\alpha\in(0,x_P)$. Then 
		$\beta_{\alpha}\in (\alpha,x_E]$, $\ga$ is
		of class $C^{2}$,  $P\in\Gamma_{\alpha}$ and
		$\ga'(x_P)=\tfrac{\partial}{\partial x} T_u(x,0)$. 
	\end{enumerate}
 \item If $\pi=1$, then $x_N=x_P$. For $\alpha\in (0,x_N)$, $(\beta_{\alpha},\ga(\beta_{\alpha}))=P$ and $\ga$ is of class $C^2$.
\end{enumerate}
   In all these cases, the function
  $G(\alpha) = \int_{\alpha}^{\beta_{\alpha}} \tfrac{\ga'(x)}{x}\, dx$ 
  has the following properties:
	\begin{equation}\label{equ:props2}
	 \begin{split}
	   \text{ $G$ is continuous on $(0,x_N)$, 
	$\textstyle\lim_{\alpha\searrow 0} G(\alpha) = +\infty$, 
	$\textstyle\lim_{\alpha\nearrow x_N} G(\alpha) = 0$. }
	 \end{split}
	\end{equation}
  In particular, $\ga$ is a solution to the free-boundary problem
  \eqref{equ:HJB-g}, \eqref{equ:integral-cond}, for some $\alpha\in
  (0,x_N)$.
\end{proposition}
\begin{remark}
  (1) 
  The parameter regime treated in Proposition \ref{pro:SECOND} above
  leads to a truly singular behavior in the ODE \eqref{equ:HJB-g}.
  Indeed, the maximal continuous solution passes through the singular
  point $P$, at which the right-hand side $L(\cdot,g)$ is not
  well-defined. It turns out that the continuity of the solution,
  coupled with the particular form \eqref{equ:HJB-g} of the equation,
  forces higher regularity (we push the proof up to $C^2$) on the
  solution. The related equation \eqref{equ:f} of Lemma \ref{lem:f}
  provides a very good model for the situation. Therein, uniqueness
  fails on one side of the equation (and general existence on the
  other), but the equation itself forces a smooth passage of any
  solution through the origin. It follows immediately, that, even
  though high regularity can be achieved at the singularity, the
  solution will never be real analytic there, except, maybe, for one
  particular value of $\log(\tfrac{1+\old}{1-\uld})$. This is a
  general feature of singular ODE with a rational right-hand sides.
  Consider, for example, the simplest case $y' = - \tfrac{y}{x^2}$
  which admits as a solution the textbook example $y(x) =
  e^{1/x}\inds{x<0}$ of
  a $C^{\infty}$ function which is not real analytic. 

  \medskip

  \noindent (2) For large-enough 
$\log(\tfrac{1+\old}{1-\uld})$, the value of $\alpha$ such that $\ga$
solves \eqref{equ:HJB-g}, \eqref{equ:integral-cond}, will fall below
$x_P$, and an interesting phenomenon will occur. Namely, the right
free boundary $\ox$ will stop depending on $\old$ or $\uld$. Indeed,
the passage through the singularity $P$ simply ``erases'' the memory
of the initial condition in $\ga$.
In financial terms, the right boundary of the
no-trade region will be stop depending on the transaction costs, while
the left boundary will continue to open up as the transaction costs
increase. 
\end{remark}
\begin{proof} 
We will only prove (1) here; (2) can be proved by the same methods used in the proof of c) below. 
For both (1) and (2), $\ga'(\alpha)=\ga(\beta_{\alpha})=0$ follows easily.\\
 \noindent	a) By Proposition \ref{Omega}, part (1), $\Omega_0 \cap \bB \cap
  \sets{(x,z)}{x>x_P} = \emptyset$. So, if $\alpha>x_P$, the statement
  can be proved by using the argument from the proof of
  Proposition \ref{pro:FIRST}, mutatis mutandis.

  \medskip
  
  \noindent b) The existence of the limit $\beta _{x_P}$ from the statement 
 is established in a matter similar to that used to prove the 
continuity of the map $\alpha
\rightarrow \ba$ in Claim 2.~in the proof of Proposition
\ref{pro:FIRST}.
The existence of the limit $g_{x_P}$ follows from a standard argument
involving a weak formulation and the dominated convergence theorem.
Finally, by part (2) of Proposition \ref{Gamma} and $T_u(x_P,0)=\tau(x_P)$, we conclude that $g_{x_P}$ is defined and continuous on $[x_P,\beta_{x_P}]$, with $(x_P,g_{x_P}(x_P))=P$.

\medskip

\noindent c) As in the proof of
Proposition \ref{pro:FIRST}, $\ga(x)$ does not hit either $\tau$ or
$T_u(\cdot,0)$ on $(0,x_P)$.
Hence, we must have $x_P\in \Cl \sD_{\alpha}$; moreover
since the curves $T_u(\cdot,0)$ and $\tau(\cdot)$ coalesce
at $x_P$, the limit $\lim_{x\to x_P} \ga(x)$ exists and equals to
$T_u(x_P,0)$. In particular, we have $x_P\in
\sD_{\alpha}$ and $P\in \Gamma_{\alpha}$.

For $x_p<x_N$, part b) above guarantees that a continuous solution with a domain
of the form $[x_P, \beta_{x_P}]$, with $\beta_{x_P}>x_P$, exists.
Therefore, by maximality, a maximal inner solution $\ga$, with
$\sD_{\alpha}=[\alpha,\beta_{x_P}]$ exists (in other words,
$\beta_{\alpha} = \beta_{x_P}$, for all $\alpha<x_P$).

Our next task is to upgrade the regularity of $\ga$ from
$C[\alpha,\beta_{\alpha}]$ to
$C^2[\alpha,\beta_{\alpha}]$, where, clearly, we can focus on a
neighborhood of the 
point $x_P$: we need to show that
that $\ga'(x_P), \ga''(x_P)$ exist and $\ga'(x), \ga''(x)$ are
continuous at $x_P$. The argument is divided in several claims, whose
proofs follow the respective statements.

\medskip

\noindent {\bf Claim 1:} {\em $\ga'(x)$ does not admit a local minimum on
  $(\alpha,x_P)\cup (x_P, \ba)$.}\ \ 
 Suppose, to the contrary, that it does. 
 Then, there exists $\eps>0$ and a point $x_m\in (\alpha,x_P)\cup
 (x_P, \ba)$ such that
 \[ \ga'(x_m)\leq \ga'(x) \textrm{  for  } x\in
   [x_m-\epsilon,x_m+\epsilon].\]
 For $k_m :=\ga'(x_m)$, parts (3) and (4) of Proposition \ref{Omega}
 imply that
 \begin{equation}\label{equ:som-label}
  \left\{\begin{split}
   \ga(x_m)&=T_u(x_m,k_m) \textrm{  or  }
	   \ga(x_m)=T_d(x_m,k_m), \\ T_d(x,k_m)&\leq \ga(x) \leq T_u(x,k_m) \textrm{
		 on  } [x_m-\epsilon, x_m + \epsilon]. 
  \end{split}\right.
 \end{equation}
We focus on the case $\ga(x_m)=T_u(x_m,k_m)$, with the other one
- when $\ga(x_m)=T_d(x_m,k_m)$  - being similar. By
\eqref{equ:som-label}, we have  $\frac{\partial}{\partial x}
T_u(x_m,k_m)=\ga'(x_m)=k_m$; moreover, 
since $\frac{\partial^2}{\partial x^2}T_u(x,k)< 0$, we get
$0= \ga'(x_m)-\tfrac{\partial}{\partial x}
T_u(x_m,k_m)<\ga'(x)-\tfrac{\partial}{\partial x} T_u(x,k_m)$, on 
$(x_m,x_m+\epsilon]$. This leads to the following
contradiction:
\[ 0<\int_{x_m}^{x_m+\epsilon}\Big(\ga'(x)-\tfrac{\partial}{\partial
	x} T_u(x,k_m)\Big)dx = \ga(x_m+\epsilon) - T_u(x_m+\epsilon,k_m)
  \leq 0.\]

\medskip

\noindent {\bf Claim 2:} {\em $\ga \in C^1([\alpha,\ba])$ and $\ga'(x)$
  decreases around $x_P$.}\ \ 
We observe that
$\tau(x)<\ga(x)<T_u(x,0)$ for $x\in (\alpha,x_P)\cup(x_P,\ba)$, 
  $\tau(x_P)=\ga(x_P)=T_u(x_P,0)$, and $
  \tau'(x_P)=\frac{\partial}{\partial x} T_u(x_P,0)$, 
and conclude that $\ga$ is differentiable at $x_P$ and
$\ga'(x_P)=\tau'(x_P)=\frac{\partial}{\partial x} T_u(x_P,0)$. By the Claim 1., $\lim_{x\nearrow x_P}\ga'(x), \lim_{x\searrow x_P}\ga'(x)$ exist. So, using the mean value theorem, we obtain $\ga'(x_P)=\lim_{x\to x_P}\ga'(x)$ and conclude that $\ga \in C^1([\alpha,\ba])$.

Given an $\eps$ in a small-enough neighborhood of $0$, 
the concavity of $T_u(\cdot,0)$ implies that
$$\ga(x_P-\epsilon)<T_u(x_P-\epsilon,0)<T_u(x_P,0) - \epsilon
\tfrac{\partial}{\partial x} T_u(x_P,0) = \ga(x_P)- \epsilon
\ga'(x_P).$$ The mean value theorem can now be used to
conclude that there exist $x_1,x_2$, arbitrarily close to $x_P$, with
$x_1<x_P <x_2$  such that
$$\ga '(x_1)>\ga '(x_P)>\ga '(x_2).$$ 
Finally, if we combine the obtained results with those of Claim 1.,  
 we can conclude that $\ga'(x)$ decreases near $x_P$.

\medskip

\noindent {\bf Claim 3:} {\em The second derivative of $\ga$ exists at $x_P$ and
}
\begin{equation}
  \label{nunu}
\begin{split}
	\ga''(x_P)=-\frac{(1-p)^2\sigma^2(2\delta-2p\mu+p(1-p)\sigma^2)^2
	(2\delta + 2(1-p)\mu + (p-2)(1-p)\sigma^2)}{p (2\delta -
	p(1-p)\sigma^2)^3}.
 \end{split}
\end{equation}
The proof is based on an explicit  computation where the easy-to-check fact that
our ODE admits the form
\[
  \ga'(x)=-\frac{(\ga(x)-T_u(x,0))(\ga(x)-T_d(x,0))}{(\ga(x)-T_u(x,\infty))(\ga(x)-T_d(x,\infty))},\]
is used. We begin with the equality
\begin{equation}
\nonumber
\begin{split}
	\frac{\ga(x)-\ga(x_P)-\ga'(x_P)(x-x_P)}{(x-x_P)^2}&=
\frac{T_d(x,\infty)-T_d(x_P,\infty)-\frac{\partial}{\partial x}
  T_d(x_P,\infty)(x-x_P)}{(x-x_P)^2}\\ & \quad -
\frac{\frac{T_d(x,\infty)-T_u(x,0)}{(x-x_P)^2}}{1+\ga'(x)
  \frac{\ga(x)-T_u(x,\infty)}{\ga(x)-T_d(x,0)}}.
 \end{split}
\end{equation}
  By L'Hospital's rule, as $x\to x_P$, the right-hand side above converges to 
  \[ 
\frac{1}{2}\frac{\partial^2}{\partial x^2} T_d(x_P,\infty) -
\frac{1}{2}\frac{\frac{\partial^2}{\partial x^2}
  T_d(x_P,\infty)-\frac{\partial^2}{\partial x^2} T_d(x_P,0)
}{1+\ga'(x_P) \frac{\ga(x_P)-T_u(x_P,\infty)}{\ga(x_P)-T_d(x_P,0)}},\]
which, in turn, evaluates to the half of the right-hand side of \eqref{nunu}.

Having computed a second-order quotient of differences for $\ga$ at
$x_P$, we could use the concavity of $\ga$ at $x_P$ (established in Claim 2.~
above) to conclude that $\ga$ is twice differentiable there. We opt to
use a short, self-contained argument, instead, where $c$ denotes the
right-hand side of \eqref{nunu}. For small enough $\zeta$, we have
\begin{displaymath}
\begin{split}
\Big(\ga'(x)-\ga'(x_P)-c (x-x_P)\Big)\zeta &\leq
\ga(x)-\ga(x-\zeta)-\ga'(x_P)\zeta - c(x-x_P) \zeta\\
&= - \tfrac{c}{2} \zeta^2 + \Big( \ga(x)-\ga(x_P)-\ga'(x_P)(x-x_P)-\tfrac{c}{2}(x-x_P)^2 \Big) \\
&\quad - \Big(\ga(x-\zeta)-\ga(x_P)-\ga'(x_P)(x-\zeta-x_P)-\tfrac{c}{2}(x-x_P-\zeta)^2 \Big) \\
&= -\tfrac{c}{2}\zeta^2 + o((x-x_P)^2) + o((x-\zeta-x_P)^2).
\end{split}
\end{displaymath}
If we fix $t>0$ and choose $\zeta=t \ \vert x-x_P \vert \
\sgn\Big( \ga'(x)-\ga'(x_P)-c(x-x_P)\Big)$, 
we obtain
$$\limsup_{x \to x_P} \Big\vert \frac{\ga'(x)-\ga'(x_P)}{x-x_P} - c
\Big \vert \leq -\tfrac{c}{2} t,$$
from which the claim follows immediately. 

\medskip

{\bf Claim 4:} {\em $\ga\in C^2([\alpha,\ba]).$}\ \ 
For convenience, we change variables as follows  
$$y=x-x_P, \quad f(y)=\ga(x)-\ga(x_P)-\ga'(x_P)(x-x_P)-\tfrac{1}{2}\ga''(x_P)(x-x_P)^2.$$
With respect to the new coordinate system, we have  $f\in C^1([\alpha-x_P,\ba-x_P]) \cap C^2([\alpha-x_P,0)\cup(0,\ba-x_P])$, and
$f(0)=f'(0)=f''(0)=0$; we need to show that $\lim_{y\to 0}f''(y)=f''(0)$.
This follows, however, directly from Lemma \ref{lem:f}, as we obtain
the ODE \eqref{equ:f} if we differentiate the equality
$g'=L(\cdot,g)$, and pass to the new coordinates. The coefficient functions
$A(y)$ and $B(y)$ admit a rather messy but explicit form which can be used
to establish their continuity. Indeed, it turns out that
$A(y)$ and $B(y)$ can be represented as 
continuous transformations of
functions of $y$, $f(y)/y^2$ and $f'(y)/y$, which are, themselves,
continuous. Similarly, the condition $A(0)>0$ imposed in Lemma \ref{lem:f} is satisfied because
one can use the aforementioned explicit expression to conclude that 
 $A(0) = \lim_{y\to 0} A(y) = 
\tfrac{(1-p)\sigma^2(A-\mu)}{2\delta -
  p(1-p)\sigma^2}>0$
\end{proof}
	\subsection{The sub-case $G\leq \mu < A$.} This sub-case is, perhaps the
most challenging of all, as it combines the existence of a singularity
with a possible failure of the well-posedness of the value function. 

For $k\in\R$ let  $l_u(k), l_d(k)$ be the (ordered) solutions $X_1,X_2$ of  the quadratic
equation $a(k) X^2 - b(k) X + c(k) = 0$, where $a(k)$, $b(k)$ and
$c(k)$ are as in \eqref{abc}.
The analysis in the sequel centers around the 
constants $C=C(\mu,\sigma,p,\mu)$ and $K=K(\mu,\sigma,p,\mu)$, given
by
\begin{equation}
\label{ite:C-expression}
K=
\tfrac{(1-p)(\mu-G)}{(A-\mu)+p(\mu-G)}
\text{ and }
  C = \int_0^{K} k \Big( \frac{l_u'(k)}{k-l_u(k)} - \frac{l_d'(k)}{k-l_d(k)} \Big) dk.
\end{equation}
\begin{lemma}
  \label{K and lim} Assume that $0<p<1$ and $G\leq \mu < A$. Then
\begin{enumerate}[leftmargin=1.8em]
\item $K$ is the smallest solution to $b(\cdot)^2 = 4 a(\cdot) c(\cdot)$.
 Moreover $K$ is nonnegative and $K=0$ if and only
  if $\mu=G$.
\item $\Omega_0 \cap \{L(x,z)=k\}$ is bounded if $k>K$ and unbounded
  otherwise.
\item For $0\leq k \leq K$, $l_d(k)>k$.
\item For $0 \leq k<K$, we have
  \[ \lim_{x\to \infty}\tfrac{\partial}{\partial x} T_{u,d}(x,k) \to
	l_{u,d}(k)\text{ and }
  \lim_{x\to \infty}\tfrac{1}{x} \tfrac{\partial}{\partial k}
  T_{u,d}(x,k) \to l_{u,d}'(k).\]
\item There exists a constant $c^*>0$ such that for $x>c^*$ and
  $k\in[0,K)$ we have
$$\Big\vert \tfrac{\frac{\partial}{\partial k} T_d(x,k)}{x(k-\frac{\partial}{\partial x} T_d(x,k))} \Big \vert <c^* + \tfrac{c^*}{\sqrt{K-k}}, \quad \Big\vert \tfrac{\frac{\partial}{\partial k} T_u(x,k)}{x(k-\frac{\partial}{\partial x} T_u(x,k))} \Big \vert <c^* + \tfrac{c^*}{\sqrt{K-k}}.$$
\item $C$ is well-defined and nonnegative. Moreover, $C=0$ if and only
  if $\mu=G$.
\end{enumerate}
\end{lemma}
\begin{proof} (1) It follows by direct computation.
%\[ b(k)^2 - 4 a(k) c(k) =4 p^2 \Big(-G^2 (p k- (1-p)) ^2 + (A k - \mu (1 +
%  k) (1 - p))^2\Big), \]
%and, form there,  that $K$ is the smaller of the two solutions to $b(\cdot)^2 - 4a(\cdot) c(\cdot)$. 

\medskip

\noindent (2) It is easily checked that the leading coefficient of
$b(k)^2-4a(k)c(k)$ (seen as a
polynomial in $k$) is positive.
%$4 p^2 \Big(((1-p)(A-\mu) + p (A-G)\Big) \Big( (1-p) (A-\mu) + p (A+G) \Big)
%  >0$. 
Therefore, $b(k)^2 - 4 a(k) c(k)\geq 0$ for $k\in [0,K]$. 
Since $b(k) - 4\delta k$ is linear in $k$ and its values at $k=0,K$ are positive, $4p(1-p)(k+1)(b(k) - 4\delta k)>0$ for $k\in [0,K]$.
%\[  
%4p(1-p)(k+1)(b(k) - 4\delta k)\vert_{k=K}= \tfrac{8p^2(1-p)^2
 % (A-\mu)(A-G)G}{(A-pG-(1-p)\mu)^2}>0, \textrm{  and  }4p(1-p)(k+1)(b(k) - 4\delta k)\vert_{k=0}=8\mu p^2 (1-p)^2>0,\]
Thus, the expression inside the square root in \eqref{T_u expression}
is positive for $x\geq 0$ and $k\in [0,K]$, which, in turn, implies
that
for $k\in [0,K]$, $\Omega_0\cap\{L(x,z)=k\}$ is unbounded. 

Similarly, since $b(k)^2-4a(k)c(k)\vert_{k=K+\epsilon}<0$ for small enough $\epsilon>0$, we conclude
that the domain $\sL_{K+\eps}$ of $T_u(\cdot,K+\epsilon)$ is bounded. Part
(4) of Proposition \ref{Omega}, implies that
$\Omega_0 \cap \{L(x,z)>K+\epsilon\}$ is a bounded set for any sufficiently small $\epsilon>0$. 
We conclude that $\Omega_0\cap\{L(x,z)=k\}$ is bounded for $k>K$.

\medskip

\noindent (3) From the definition of $l_d(k)$ we get
\[2 a(k) \Big(l_d(k)-k\Big)=b(k)-4\delta k + 4p\delta \, k(\tfrac{1-p}{p}-k)
-\sqrt{b(k)^2 - 4a(k)c(k)}\]
We already checked that $b(k)-4\delta k>0$  for $k\in[0,K]$. Also, $\tfrac{1-p}{p}-k>0$ for $k\in[0,K]$, since $\tfrac{1-p}{p}-K = \tfrac{1-p}{p} \cdot \tfrac{A-\mu}{A-\mu +p(\mu-G)}>0$. Thus, $b(k)-4\delta k + 4p\delta k(\tfrac{1-p}{p}-k)>0$ for $k\in [0,K]$. Furthermore, 
\begin{multline}
  \label{mli}
  \Big(b(k)-4\delta k + 4p\delta k(\tfrac{1-p}{p}-k)\Big)^2 -
\Big(b(k)^2 - 4a(k)c(k)\Big) =  \\
=8p^2\delta (1+k) (\tfrac{1-p}{p}-k) (-2\delta k^2 + p(2\mu-\sigma^2)k +p^2 \sigma^2).
\end{multline}
We can now conclude that the left-hand side of \eqref{mli} is positive
on $[0,K]$, since
the function  $k\mapsto (-2\delta k^2 + p(2\mu-\sigma^2)k +p^2 \sigma^2)$ is concave and its 
values at $k=0, K$ are positive. It follows immediately that
$l_d(k)>k$ for $k\in [0,K]$.

\medskip

\noindent (4) This can be shown by the direct computation.

\medskip

\noindent (5) A straightforward (but somewhat tedious) calculation yields that
$\tfrac{\partial}{\partial x} T_d(x,k) \to l_d(k)$, as $x\to\infty$,
uniformly in $k\in [0,K]$.
So, by (3), we can choose $c^*$ such that $\tfrac{\partial}{\partial x} T_d(x,k)-k>\epsilon$ for some $\epsilon>0$ and all $x>c^*$, $k\in [0,K]$.
Also, we can check that there exists a constant $c^*$ such that
for $x>1$ and $k\in [0,K)$ we have
$$\tfrac{1}{x}\tfrac{\partial}{\partial k} T_d(x,k)< c^*+ c^* \tfrac{1}{\sqrt{b(k)^2-4a(k)c(k)}} <c^*+ c^* \tfrac{1}{\sqrt{K-k}},$$
whence the first inequality in the statement of (5) follows. The
second one is obtained in a similar manner.

\medskip

\noindent (6) We first observe that $l_d'(k)>0$ and $l_u'(k)<0$ for $k\in[0,K)$. Then, 
the statement follows from the integrability of $1/\sqrt{K-\cdot}$ on $[0,K]$
and the fact that
$\Big\vert
k(\tfrac{l_u'(k)}{k-l_u(k)}-\tfrac{l_d'(k)}{k-l_d(k)})\Big\vert <c^*+
c^* \tfrac{1}{\sqrt{K-k}}$, which is, in turn, implied by (4) and (5)
above.
\end{proof}

\begin{remark}\label{rem:K} In our current parameter range ($0<p<1$,
  $G\leq \mu < A$), 
  the level curve $L=0$ is a hyperbola and the curve $L=k$ is an
  ellipse for large-enough values of $k$. In fact,
 $K$ is the smallest value of $k\geq 0$ such that $L=k$ is a hyperbola (and, therefore, unbounded). 
\end{remark}

For $G\leq \mu < A$, the Merton proportion $\pi$ cannot take the value $1$, so we only consider the
cases $\pi<1$ and $\pi>1$ in the following proposition:
\begin{proposition} [$0<p<1$, $G\leq \mu < A$]
  \label{pro:THIRD}
Assuming that $0<p<1$  and $G\leq \mu < A$, we have the following
statements:
\begin{enumerate}[leftmargin=1.8em]
  \item If $\pi < 1$, then $P \not\in\Gamma_{\alpha}$, for each $\alpha>0$.
\item If $\pi > 1$ then $P\in\Gamma_{\alpha}$ if and only if
  $\alpha\leq x_p$.
\end{enumerate}
In both cases, $\beta_{\alpha}<\infty$. 
Moreover, for $G(\alpha)=\int_{\alpha}^{\ba} \tfrac{\ga'(x)}{x}\, dx$,
we have
\
\begin{equation}
  \label{limit-integral}
  \text{ $G$ is continuous on $(0,\infty)$, }
  \textstyle \lim_{\alpha\searrow 0} G(\alpha) = +\infty,\text{ and }
  \lim_{\alpha\nearrow +\infty} G(\alpha) = C,
\end{equation}
 where $C$ is given by \eqref{ite:C-expression}.

\end{proposition}

\begin{proof} 
  The parts of statements (1) and (2) involving singularities are proved similarly to parallel
  statements in Proposition \ref{pro:SECOND}. We show that
  $\beta_{\alpha}<\infty$ for $\pi<1$, with the case $\pi>1$ being
  quite similar. Proceeding by contradiction, we 
  suppose that $\beta_{\alpha}=\infty$, for some $\alpha>0$. Then,
  just like in the proof of Proposition \ref{pro:SECOND}, we can show
  that $\ga'(x)$ does not
  admit a local minimum on $(\alpha,\infty)$. Thus, there exists $k^*$
  such that $\lim_{x\to \infty} \ga'(x)=k^*$. From Proposition
  \ref{K and lim}, part (2), we learn that $k^*\in [0,K]$, whereas
  from
  part (3) 
  we conclude that there exists $\epsilon>0$ such that
  $l_d(k^*-\epsilon)>k^* + 2\epsilon$. 
Since $\vert \ga'(x)-k^*\vert<\epsilon$ for large enough $x$, we can
use part (4) of 
Proposition \ref{K and lim}, to obtain a contradiction
 $$\lim_{x\to\infty}
\Big(\tfrac{\partial}{\partial x}T_d(x,k^*-\epsilon)-\ga'(x)\Big)>k^*
+ 2\epsilon - (k^*+\epsilon)=\epsilon,$$ with the fact that
the inequality 
$\ga'(x)>k^*-\epsilon$ implies that $\ga(x)>T_d(x,k^*-\epsilon)$, for
large $x$.

\noindent \begin{minipage}{0.7\textwidth}
  It remains to prove \eqref{limit-integral}. 
 The main idea is to 
intersect the solution $\ga$ with the (unbounded) level curve $L=K$.
If the two points of intersection are denoted by $x_u$ (the intersection is on $T_u(\cdot, K)$) and $x_d$ (intersection on  $T_d(\cdot, K)$), with $x_u<x_d$ (see Figure 8), then the integral in 
 \eqref{limit-integral} is split into three integrals on the intervals
	 $[\alpha,x_u]$, $[x_u,x_d]$ and $[x_d,\ba]$.
The first and the last integrals are then computed using the change of
variable $k=\ga '(x)$, while the  limit of the middle integral is
shown to be zero.
\end{minipage}
\hfill
\begin{minipage}{0.27\textwidth}
\begin{center}
\includegraphics[width=4.5cm]{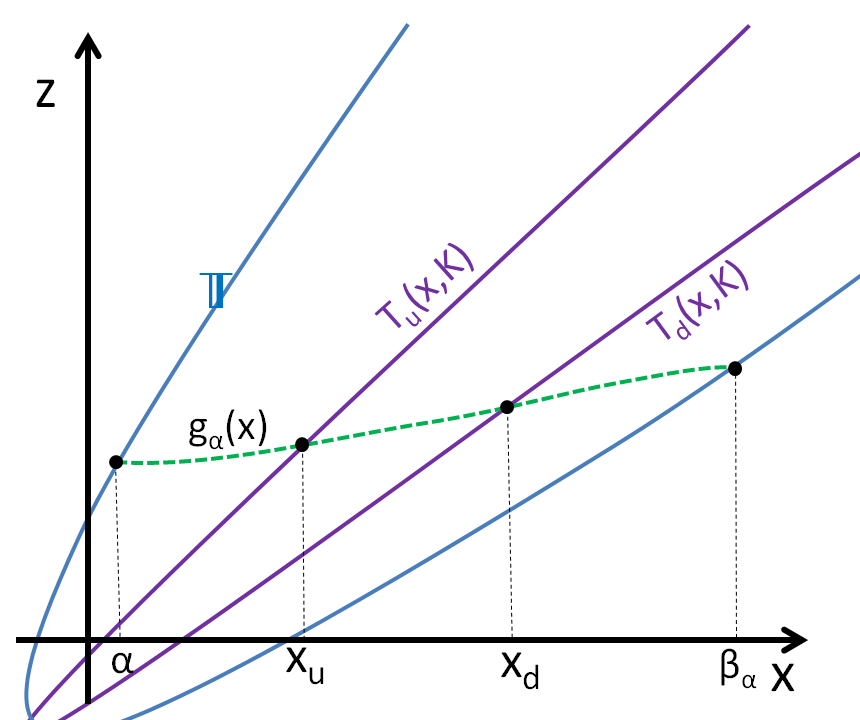}
\end{center}
\figcaption{$x_u$ and  $x_d$}
\end{minipage}

We start this program by observing that 
the region
$\Omega_0 \cap \{L(x,z)=K\}$ is unbounded (see Proposition
\ref{K and lim} (2)), and, hence, so is the region $\Omega_0\cap \{
  L(x,z)>K \}$. Also, we observe that $T_u(x,0)>T_u(x,K)>T_d(x,K)>T_d(x,0)$ for $x\in (0,\infty)$.
We conclude from there  that 
$\Gamma_{\alpha}$ intersects the region
$$\Omega_0 \cap \{(x,z): T_d(x,K)<z<T_u(x,K)\}=\Omega_0 \cap
\{L(x,z)>K\}.$$
Therefore, 
$\ga'(x_m(\alpha))>K\text{ for }x_m(\alpha)\in \argmax_{x\in
  [\alpha,\ba]} \ga'(x)$.

Since $\ga'(x)$ doesn't admit a local minimum on $(\alpha,\ba)$,
$x_m(\alpha)$ is uniquely defined and 
$\ga'(x)$ strictly increases on $(\alpha,x_m)$ and strictly decreases on
$(x_m,\ba)$. Consequently, 
there exists a pair $x_u(\alpha), x_d(\alpha)$ with  $x_u(\alpha)\in (\alpha,x_m(\alpha))$ and
$x_d(\alpha)\in (x_m(\alpha),\ba)$ such that 
\[ \ga'(x_u(\alpha))=K, \quad \ga'(x_d(\alpha))=K.\]
Let $I_{\alpha} : [0,K]\mapsto
[\alpha,x_u(\alpha)]$ be the  inverse function of $\ga'(x)$ on
$[\alpha,x_u(\alpha)]$, so that
\[\ga'(I_{\alpha}(k))=k,\quad
\ga(I_{\alpha}(k))=T_u(I_{\alpha}(k),k)\text{ and }
I_{\alpha}'(k)= \tfrac{ \frac{\partial}{\partial
k}T_u(I_{\alpha}(k),k)}{ k - \frac{\partial}{\partial
x}T_u(I_{\alpha}(k),k)},\]
where the last equality can be obtained by differentiating the middle
one. 
A change of variables $x=I_{\alpha}(k)$ yields
\begin{equation}\label{limit1}
\int_{\alpha }^{x_u(\alpha)} \tfrac{g_{\alpha}'(x)}{x} dx = \int_{0}^{K} \tfrac{k}{I_{\alpha}(k)} \tfrac{ \frac{\partial}{\partial k}T_u(I_{\alpha}(k),k)}{ k - \frac{\partial}{\partial x}T_u(I_{\alpha}(k),k)} dk \quad  \stackrel{\alpha \to \infty}{\longrightarrow}\quad  \int_{0}^{K} \tfrac{k l_u'(k)}{k-l_u(k)}dk,
\end{equation}
where the existence of the limit and its value are obtained using
parts (4) and (5) of Proposition \ref{K and lim}, 
together with the fact that
$\lim_{\alpha \to \infty} I_{\alpha}(k) = \infty$.
In particular, part (5) of Proposition \ref{K and lim}  allows us to apply the dominated convergence theorem.
Similarly, we have
\begin{equation}\label{limit2}
 \int_{x_d(\alpha)}^{\ba} \tfrac{g_{\alpha}'(x)}{x} dx  \quad  \stackrel{\alpha \to \infty}{\longrightarrow}\quad  
- \int_{0}^{K} \tfrac{k l_d'(k)}{k-l_d(k)}dk.
\end{equation}

\medskip

\noindent It remains to show that
$\int_{x_u(\alpha)}^{x_d(\alpha)} \frac{\ga'(x)}{x}dt \to 0$ as
$\alpha \to \infty$.
By Proposition \ref{K and lim} (parts (3) and (4)), 
there exist $\epsilon>0$ and $x_{\epsilon}>0$ such that
$\tfrac{\partial}{\partial x}T_d(x,K)>K+2\epsilon$, for
$x>x_{\epsilon}$.
Moreover, part (2) of the same proposition guarantees the existence
of a constant $\alpha_{\epsilon}>0$ such that 
$$\Omega_0 \cap \{L(x,y)>K+\epsilon\} \subset \{x \leq \alpha_{\epsilon}\}.$$
Then, $\ga'(x)<K+\epsilon$ for $\alpha>\alpha_{\epsilon}$ and $x\in
[\alpha,\ba]$, and, so, for  $\alpha>\alpha_{\epsilon} \vee
x_{\epsilon}$, we have
\begin{equation}
  \nonumber 
  \begin{split}
\tfrac{1-p}{\delta} \sqrt{1+ \tfrac{b(K)-4\delta K}{p(1-p)(1+K)}x_u(\alpha) } &= T_u(x_u(\alpha),K)-T_d(x_u(\alpha),K) \\
&=\ga(x_u(\alpha))-T_d(x_u(\alpha),K) + T_d(x_d(\alpha),K)- \ga(x_d(\alpha))\\
&=\int_{x_u(\alpha)}^{x_d(\alpha)} \Big(\tfrac{\partial}{\partial x}T_d(x,K)-\ga'(x)\Big) dx  
\geq  \epsilon \ (x_d(\alpha)-x_u(\alpha)),
 \end{split}
\end{equation}
where the first equality follows by direct computation, the
second one by the fact that
$\ga(x_u(\alpha))=T_u(x_u(\alpha),K)$ and
$\ga(x_d(\alpha))=T_d(x_d(\alpha),K)$, and the final
inequality from the choice of $\alpha$. Hence, 
\begin{equation}\nonumber
 \begin{split}
\limsup_{\alpha \to \infty}\Big\vert \int_{x_u(\alpha)}^{x_d(\alpha)}
\tfrac{\ga'(x)}{x} dx \Big\vert &\leq \limsup_{\alpha \to \infty}\Big\vert (K+\epsilon) \ln{\Big( 1+ \tfrac{x_d(\alpha)-x_u(\alpha)}{x_u(\alpha)} \Big)}\Big\vert \\
&\leq \limsup_{\alpha \to \infty}\Big\vert (K+\epsilon) \ln{\Big( 1+
  \tfrac{1-p}{\epsilon  \delta\ x_u(\alpha)} \sqrt{1+ \tfrac{b(K)-4\delta K}{p(1-p)(1+K)}
	x_u(\alpha) } \Big)}\Big\vert =0. \qedhere
 \end{split}
\end{equation}
\end{proof}

The remaining task in the proof of Theorem \ref{thm:well-posed} is to
show that the problem is not well posed, whenever $\log
(\frac{1+\old}{1-\uld})\leq C$.
\begin{proposition}\label{prop:wellposed} Assume that
$p\in (0,1)$ and   $G \leq \mu < A$. 
 If $\log(\tfrac{1+\old}{1-\uld})\leq C$, where $C$ is defined in
 \eqref{ite:C-expression},  then $u=\infty$, i.e., the
 problem is not well posed. 
\end{proposition}
\begin{proof}
Without loss of generality, we assume that $\uld=0$; indeed, it is
enough to scale (the initial value of) the stock price $\prfi{S_t}$ by $(1-\uld)$, otherwise. 

  For $\alpha>0$, the function $\ga$ in Proposition \ref{pro:THIRD}
corresponds to
the value function $u$ under the transaction costs $\old$ and $\uld=0$
such that
$G(\alpha) = \log(1+\old)$, where $G(\alpha)=\int_{\alpha }^{\beta _{\alpha}} \frac{g_{\alpha}'(x)}{x}
dx$.
More precisely, Lemma \ref{lem:complete-optimal} in Section 4~above 
yields that
\begin{equation}\label{representation-ga}
u(1, 0, \old,0)=\tfrac{1}{p}\ga
(\alpha)^{1-p}=\tfrac{1}{p}T_u(\alpha)^{1-p}, \textrm{~if~}
G(\alpha) =\log(1+\old),
\end{equation}
where $u(\eta _B, \eta _S, \old, \uld)$ is the optimal utility
for the initial position $\eta _B, \eta _S$,
under the transaction costs $\old$ and $\uld$.
The strict increase of $T_u$ and 
the decrease of $u(1,0,\cdot, 0)$, imply that $G(\alpha)$
is strictly decreasing, wherever it is defined. It now easily follows that
\[\lim _{\alpha \nearrow \infty}\tfrac{1}{p}\ga (\alpha ) ^{1-p}=\lim
_{\alpha \nearrow \infty}\tfrac{1}{p}T_u(\alpha)^{1-p}=\infty,\]
which, together with \eqref{limit-integral} and the representation
\eqref{representation-ga}, yields that
$\lim _{\log (1+\old)\searrow C} u(1,0, \old,
0)=\infty$. 
Since, clearly, $u(1,0, \cdot, 0)$ is decreasing in $\old$, this amounts
to saying that
\[ u(1,0,\old, 0)=\infty,\textrm{~for~}\log (1+\old)\leq C.\qedhere\]
\end{proof}
\begin{remark}
The map 
$\alpha \rightarrow  G(\alpha)=\int_{\alpha }^{\beta _{\alpha}}
\frac{g_{\alpha}'(x)}{x} dx$
is strictly decreasing in general, not just under the parameters
restricted by the hypothesis of Proposition \ref{prop:wellposed}.
The same argument, as the one given in the proof of Proposition
\ref{prop:wellposed}, applies. In particular, this fact can be used to show
that the free-boundary problem \eqref{equ:HJB-g},
\eqref{equ:integral-cond} has a \emph{unique} solution for all values
of the transaction costs, as long as $u<\infty$. 

It is, perhaps, interesting to note that the authors are unable to
come up with a  purely analytic argument
for the monotonicity of $G(\alpha)$. The crucial step in the proof of
Proposition \ref{prop:wellposed} above is to relate the value of
$G(\alpha)$ to the original control problem, and then argue by using
the natural monotonicity properties of the control problem itself,
rather than the analytic description \eqref{equ:HJB-g} only.
\end{remark}

\providecommand{\bysame}{\leavevmode\hbox to3em{\hrulefill}\thinspace}
\providecommand{\MR}{\relax\ifhmode\unskip\space\fi MR }
\providecommand{\MRhref}[2]{%
  \href{http://www.ams.org/mathscinet-getitem?mr=#1}{#2}
}
\providecommand{\href}[2]{#2}

\end{document}